\documentclass{lmcs}
\pdfoutput=1

\usepackage{lastpage}
\lmcsdoi{18}{3}{28}
\lmcsheading{}{\pageref{LastPage}}{}{}%
{Nov.~30,~2020}{Sep.~07,~2022}{}

\usepackage[utf8]{inputenc}

\input{lindep-long.sty}

\begin{document}

\title{Linear Dependent Type Theory for Quantum Programming Languages}
  
\author{Peng Fu\rsuper{a}\lmcsorcid{0000-0002-3123-0867}}
\author{Kohei Kishida\rsuper{b}\lmcsorcid{0000-0002-6719-1521}}
\author{Peter Selinger\rsuper{c}\lmcsorcid{0000-0003-3161-856X}}

\address{Dalhousie University, Canada}
\email{frank-fu@dal.ca, selinger@mathstat.dal.ca} 

\address{University of Illinois at Urbana-Champaign, U.S.A.}
\email{kkishida@illinois.edu}



\begin{abstract}
  Modern quantum programming languages integrate quantum resources and
  classical control.  They must, on the one hand, be linearly typed to
  reflect the no-cloning property of quantum resources.  On the other
  hand, high-level and practical languages should also support quantum
  circuits as first-class citizens, as well as families of circuits
  that are indexed by some classical parameters.  Quantum programming
  languages thus need linear dependent type theory.  This paper defines
  a general semantic structure for such a type theory via certain
  fibrations of monoidal categories.  The categorical model of the
  quantum circuit description language Proto-Quipper-M in
  {\cite{RiosS17}} constitutes an example of such a fibration, which
  means that the language can readily be integrated with dependent
  types.  We then devise both a general linear dependent type system and
  a dependently typed extension of Proto-Quipper-M, and provide them
  with operational semantics as well as a prototype implementation.
\end{abstract}

\maketitle

\section{Introduction}

\subsection{Quantum programming and linear types}

Quipper \cite{green2013quipper,green2013introduction} is a functional programming language for describing and generating quantum circuits.
In addition to providing a low-level paradigm for generating circuits, where circuits are constructed by applying one gate at a time, it also provides a high-level paradigm, where circuits are first-class citizens to which meta-operations can be applied.
This was found to provide an appropriate level of abstraction for formalizing many quantum algorithms in a way that is close to how they are informally described in the literature {\cite{VRSAS2015-cacm}}.

Quipper is implemented as an embedded language within Haskell and lacks formal foundations.
This is why a family of languages known as Proto-Quipper \cite{ross2015algebraic,RiosS17} has been developed as a formalization of suitable fragments of Quipper, susceptible to formal methods.

One aspect of the formal foundations missing in Quipper is linear types.
The ``no-cloning'' property of quantum mechanics means that quantum resources are linear in the sense that one cannot duplicate a quantum state \cite{nielsen2002quantum}.
This is why linear types \cite{girard1987linear,wadler1990linear} have been adopted to provide a type theory for quantum computing \cite{selinger2005lambda}.
Since Quipper lacks linear types,
  it leaves the programmer with the responsibility to keep track of the use of non-duplicable quantum resources.
Proto-Quipper solves this problem by giving a syntactically sound linear type system, so that any duplication of a quantum state will be detected as a typing error at compile time.

Since Quipper is a circuit description language, it shares certain aspects of hardware description languages, such as a distinction between \emph{circuit generation time} and \emph{circuit runtime}.
Values that are known at circuit generation time are called \emph{parameters} and values that are known at circuit runtime are called \emph{states}. 
The parameter/state distinction ties in with linearity, because although states can be linear (such as qubits) or nonlinear (such as classical bits used in circuits), parameters are always classical and therefore never linear. 
One of Quipper's design choices {\cite{VRSAS2015-cacm}} is that parameters and states share the same name space; in fact, they can occur together in a single data structure, such as a pair of an integer and a qubit, or a list of qubits (in which case the length of the list is a parameter and the actual qubits in the list are states).

Among the languages of the Proto-Quipper family, Proto-Quipper-M \cite{RiosS17} has several desirable features.
Its linear type system accommodates nontrivial interactions of parameters and states.
It supports higher-order functions and quantum data types such as lists of qubits.
And, most significantly for our purposes, it has a denotational semantics in terms of a categorical model, called $\mbbar$, that takes advantage of a linear-nonlinear adjunction in the sense of Benton \cite{benton1994mixed}.

\subsection{Quantum programming and dependent types}

The existing versions of Proto-Quipper do not support dependent types \cite{martin1984intuitionistic}.

Dependent types are good at expressing program invariants and constraints at the level of types \cite{nordstrom1990programming}.
In the context of quantum circuit programming, dependent types give the programmer access to \emph{dependent quantum data types}.
For example, when one considers a linear function such as the quantum Fourier transform, it is really a family of circuits, rather than a single circuit, indexed by the length of the vector of qubits that the function takes as input.
Another important application of dependent types in quantum circuit programming is the use of existential dependent quantum data types to hide garbage qubits \cite{MoscaRS18,FKRS2020}, while guaranteeing that they will be uncomputed.  
This facilitates the programming of large scale reversible circuits.
It is therefore desirable to equip Proto-Quipper with dependent types.

The problem of how to formally add dependent types to a quantum programming language is not restricted to Proto-Quipper.
The quantum circuit description language QWire \cite{paykin2017qwire} also uses a linear type system and has a denotational semantics (based on density matrices). 
It was recently embedded in the proof assistant Coq \cite{rand2018qwire}, making it possible to formally verify some simple quantum programs.
QWire does not support quantum data types, much less dependent ones.
Paykin et al.\@ \cite[Section 6.2]{paykin2017qwire} briefly discuss dependent types in the context of QWire 
but neither a semantic model nor a detailed analysis is given.
The problem of defining a denotational semantics that can accommodate linear and dependent types for quantum circuit programming has remained open until now.

\subsection{Linear dependent types}

Designing a linear dependent type system is a challenge.
One of the major hurdles is to understand what it means for a type to depend on a linear resource, and to provide a semantics for that kind of dependency.
Indeed, suppose that a function $f$ has a linear dependent type $(x : A) \multimap B[x]$ and that we are given a linear resource $a:A$. 
How should we understand $f a : B[a]$ (note that this seems to use $a$ twice, once in a term and once in a type)?

Several approaches to linear dependent types have been proposed.
Cervesato and Pfenning \cite{cervesato1996linear} extended the intuitionistic logical framework LF with linear types of the form $A \multimap B$.
Their system separates a context into a linear part and an intuitionistic part:
On the one hand, to form a lambda term $\hat{\lambda} x . M : A \multimap B$, the context $x \mathrel{\hat{:}} A$ must be a linear one.
On the other hand, to form a lambda term $\lambda x. M : (x : A) \to B[x]$ of a dependent type, the context $x : A$ must be an intuitionistic one.
When applying the linear lambda term $\hat{\lambda} x . M$ to another term, one combines two linear contexts as usual.
On the other hand, when applying the intuitionistic lambda term $\lambda x . M$ to another term $N$, one has to make sure that $N$ does not contain any linear variables.
Note that in this system, linear function types and dependent function types are completely separate; effectively this means that parameters and states live in different name spaces.
A similar approach is taken by Krishnaswami et al.~\cite{krishnaswami2015integrating}, who proposed a type system extending Benton's logic LNL \cite{benton1994mixed},
V\'ak\'ar~\cite{vakar15}, who extended Barber's dual intuitionistic linear logic \cite{barber1996dual} with dependent types, and Gaboardi et al.~\cite{gaboardi2013linear,DalLagoGaboardi2011}, who combined lightweight indexed types with bounded linear types \cite{girard1992bounded}.  

From a unifying perspective, McBride \cite{mcbride2016got} proposed the use of indices $0$, $1$, $\omega$ to decorate the context.
The indices $0$, $1$, and $\omega$ mean that the variable they decorate is, respectively, never used, used exactly once, and used more than once.
He then introduced linear dependent types of the form $(x :_k A) \multimap B[x]$, which correspond to the irrelevant quantification $\forall (x : A). B[x]$ \`a la Miquel \cite{miquel2001implicit} when $k = 0$, to Cervesato and Pfenning's dependent type when $k = \omega$, and to the linear dependent type when $k = 1$. The typing judgments in McBride's system are of the
form $\Gamma \vdash M :_k A$, where the index $k$ on the right of turnstile separates
the set of typing judgments into two kinds: $\Gamma \vdash M :_0 A$ for terms that occur in types and
$\Gamma \vdash M :_1 A$ for terms that will be evaluated at runtime.   
McBride's understanding of $fa :_1 B[a]$ is that the term $a$ that occurs in the type $B[a]$ has been ``contemplated'' (which happens at type checking time) but should not be considered to have been ``consumed'' for the purpose of linearity (which can only happen at runtime).

In this paper, we propose to interpret the dependence of $B$ on $x$ in $(x : A) \multimap B[x]$ as saying that $B$ depends on the \textit{shape} of $x$.
Informally, if $N$ is a term of a quantum data type $A$, for example a tree whose leaves are qubits, then the shape of $N$, denoted $\Sh(N)$, is the same tree, but without the leaves (or more precisely, where the data in the leaves has been replaced by data of unit type).
We write $\Sh(A)$ for the type of shapes of $A$, so that $\Sh(N):\Sh(A)$.
With these conventions, the dependence of $B$ on $x$ in $(x : A) \multimap B[x]$ is interpreted as  $x : \Sh(A) \vdash B[x] : *$.
It is understood that, for any resource $N$ (linear or otherwise), its shape $\Sh(N)$ is duplicable. 
Accordingly, our typing rule for application has the following form (where $\Gamma_i$ has the form $x_1 :_{k_1} A_1, \ldots, x_n :_{k_n} A_n$, following McBride).
\[
  \infer{\Gamma_1 + \Gamma_2 \vdash M N : B[\Sh(N)]}{\Gamma_1 \vdash M : (x : A) \multimap B[x] &
    \Gamma_2 \vdash N : A}
\]
So instead of separating typing contexts into linear and nonlinear parts like
Cervesato and Pfenning, or of having two kinds of typing judgments like McBride and like
Atkey \cite{atkey2018syntax}, we identify a subset of typing judgments of the
form $\Sh(\Gamma) \vdash \Sh(N) : \Sh(A)$. This subfragment
of typing accounts for the (duplicable) intuitionistic fragment that does not change
the quantum states. In this way, the notion of shape does not only play a conceptual role in our understanding of the dependence of $B$ on $A$;
it is indeed part of the formalism of our treatment of linear dependent types.
Since the intuitionistic fragment $\Sh(\Gamma) \vdash \Sh(N) : \Sh(A)$
is only used for typing and kinding purposes, in practice
programmers do not directly program with terms such as $\Sh(N)$.

A related concept under the name \textit{shapely types}
was considered by Jay and Cockett~\cite{jay1994shapely}. Their notion of \textit{shapely types}
is defined as certain pullbacks of the list functor. Although our notion
of \textit{shape-unit functor} is defined differently, it appears both notions of shapes
coincide for data types such as lists. However, our notion of shape-unit functor or the shape
operation can also be used for mapping a state-modifying program into a ``pure'' program that
does not modify state. This is beyond Jay and Cockett's treatment of shapely types.

\subsection{Contributions}

In this paper, we provide a new framework of linear dependent type theory suitable for quantum circuit programming languages.
The framework comes with both operational and denotational semantics.
Indeed, the categorical structure for our denotational semantics subsumes Rios and Selinger's category $\mbbar$, showing that Proto-Quipper-M can be extended with dependent types.
Furthermore, our categorical structure, in terms of a fibration of a symmetric monoidal closed category over a locally cartesian closed category, can be viewed as a generalization of Seely's \cite{seely1984locally} semantics of nonlinear dependent type theory.

\begin{itemize}
  \item In Section~\ref{model}, we define a notion of state-parameter fibration.
	We then show how the state-parameter fibration naturally gives
	rise to concepts such as
	dependent monoidal product, dependent function space, and
	a \emph{shape-unit functor} that interprets the shapes of types and terms.

  \item In Section~\ref{abstract}, we provide typing rules and an operational semantics
	for a linear dependent type system
	that features the \textit{shape operation}.
	We show that the shape operation corresponds to the shape-unit functor on
	the state-parameter fibration.
	We interpret the type system using the state-parameter fibration
	and prove the soundness of the interpretation.

  \item In Section~\ref{quantum}, we show how to extend
	the linear dependent type system of Section~\ref{abstract}
	to support programming quantum circuits, giving rise to a
	dependently typed version of Proto-Quipper. We discuss
	some practical aspects of the implementation.
\end{itemize}

This paper is an expanded version of a paper presented
at LICS 2020 {\cite{fu2020linear}}.
This paper has added proofs for Section~\ref{model}, a proof
sketch for the soundness theorem (Theorem~\ref{soundness}) and a
fuller exposition of the rules in Section~\ref{abstract}.
It also includes an appendix that contains more details of
Theorem~\ref{interp} and a full description of dependently typed
Proto-Quipper.

\section{Categorical structure}
\label{model}

Since Seely \cite{seely1984locally}, it has been well understood how (nonlinear) dependent type theory can be interpreted in locally cartesian closed categories (LCCCs).
To interpret a linear dependent type theory, one might naively look for an analogous notion of ``locally monoidal closed category.''
However, no such structure is readily available.
Instead, to build our model, we must take the notions of parameter and state seriously. 
We believe that the correct way to do this is to consider a fibration of a symmetric monoidal closed category (representing states) over an LCCC (representing parameters).

\subsection{Pullbacks and fibrations}

In modelling dependent types, pullbacks play the essential role of interpreting substitution of terms into types and contexts.
To carry this idea over to the linear setting, the perspective of fibrations \cite[Ch.\ 8]{borceux1994handbook} proves useful.

\begin{definition}\label{def:model.pullbacks.cartesian}
  Given a functor $F : \EE \to \BB$, a morphism $f : B \to A$ of $\EE$ is said to be \emph{cartesian} if
  \begin{itemize}
	\item
	  For every $g : C \to A$ in $\EE$ and $b : F(C) \to F(B)$ in $\BB$ such that $F(g) = F(f) \circ b$, there is a unique $u : C \to B$ such that $F(u) = b$ and $f \circ u = g$.
	  \[
      \begin{tikzcd}
        C \arrow[dr, dashed, "u"] \arrow[drr, bend left = 10, "g"]
        \arrow[d, mapsto, "F"] & & \\
        F(C) \arrow[dr, "b", swap]
		& B \arrow[d, mapsto, "F"] \arrow[r, "f"]& A \arrow[d, mapsto, "F"]\\
        & F(B) \arrow[r, "F(f)", swap] & F(A)
      \end{tikzcd}
	  \]
  \end{itemize}
  Given an arrow $a : X\to F(A)$ of $\BB$, we say that $f:B\to A$ is a \emph{cartesian lifting} of $a$ if $a=F(f)$ and $f$ is cartesian.
  We say that $F$ is a \emph{Grothendieck fibration}, or \emph{fibration} for short, if every arrow $a : X \to F(A)$ has a cartesian lifting at $A$.
\end{definition}

As the picture above suggests, cartesian arrows are akin to pullbacks.
In fact, we have the following:

\begin{fact}
  \label{pullback-cartesian}
  Suppose there is an adjoint pair $F \dashv G$ (where $F : \EE \to \BB$).
  Let $\eta$ be the unit of the adjunction.
  Then the left square is a pullback if
  and only if $f$ is the cartesian lifting of $F(f)$ on the right.
  \[
	\begin{tikzcd}
	  B \arrow[d, "\eta_B", swap] \arrow[r, "f"]
	  \arrow[dr, phantom, "\LRcorner", very near start] & A
	  \arrow[d, "\eta_A"]
	  & B \arrow[d, mapsto, "F"]\arrow[r, "f"] & A \arrow[d, mapsto, "F"] \\
	  GF(B) \arrow[r, "GF(f)"] & GF(A) & F(B) \arrow[r, "F(f)"]& F(A)
	\end{tikzcd}
  \]
\end{fact}
\begin{proof}
  First let us observe the following.
  Let $k : F(C) \to F(B)$ and $h : C \to GF(B)$ be morphisms such that $\tilde{h} = k$ and $\hat{k} = h$. The triangle below commutes if and only if $k = F(g)$.
  \[
    \begin{tikzcd}
      C \arrow[rd, "g"] \arrow[rdd, bend right, "h"] &    &  & C \arrow[rd, "g"] &  \\
      & B \arrow[d, "\eta_{B}"] &  & F(C) \arrow[rd, "F(g)",swap] & B   \\
      & GF(B)   &  &   & F(B)
    \end{tikzcd}
  \]
  This is because

  \begin{itemize}
	\item If $h = \eta_B \circ g$, then $k = \tilde{h} = \widetilde{\eta_B \circ g} = F(g)$.
	  The last equality is by the uniqueness of the $\widetilde{(-)}$ operation and the naturality of $\eta$.
	\item If $k = F(g)$, then $h = \hat{k} = \widehat{F(g)} = \eta_B \circ g$.
	  The last equality is by the uniqueness of the $\widehat{(-)}$ operation and the unit-counit identity property.
  \end{itemize}

  Now, in the statement of the Fact~\ref{pullback-cartesian},
  suppose the left square is a pullback. To show that $f$ is a cartesian lifting of $F(f)$,
  let $g : C \to A$ and $k : F(C) \to F(B)$ be such that $F(g) =  F(f) \circ k$. Then
  $GF(f) \circ \hat{k} = \widehat{F(f)\circ k} = \widehat{F(g)} = \eta_A \circ g$.
  Therefore the pullback gives a unique $u : C \to B$ such that $f \circ u = g$ and
  $\eta_B \circ u = \hat{k}$, the latter of which is equivalent
  to $F(u)=  k$ by the observation.

  Suppose $f$ is a cartesian lifting of $F(f)$ . To show that the left square is a pullback,
  let $g : C \to A$ and
  $h : C \to GF(B)$ be such that $\eta_A \circ g = GF(f) \circ h$.
  This implies (by the observation) that $F(g) = \widetilde{GF(f) \circ h} = F(f) \circ \tilde{h}$. Therefore the cartesian lifting $f$ gives a unique $u : C \to B$ such that
  $f \circ u = g$ and $F(u) =\tilde{h}$, the latter
  of which is equivalent to $\eta_B \circ u =  h$ by the observation.
\end{proof}

The following subclass of fibrations forms the basis of our categorical semantics for linear dependent type theory.%
\footnote{For the sake of convenience, we use a single enumerated list \eqref{pullback-fibration.1}--\eqref{state-parameter.3} across Definition~\ref{pullback-fib}, Fact~\ref{term-fibered}, and Section~\ref{fib:param-state}, so that we may refer to these items later.}

\begin{definition}
  \label{pullback-fib}
  Let us call a Grothendieck fibration $\sharp : \EE \to \BB$ a \emph{pullback fibration} if
  \begin{enumerate}
	  \setcounter{enumi}{\value{equation}}
	\item
	  \label{pullback-fibration.1}
	  $\BB$ has pullbacks, and
	\item
	  \label{pullback-fibration.2}
	  $\BB$ and $\EE$ each have a terminal object $1$ and
          $\sharp$ preserves it.  \setcounter{equation}{\value{enumi}}
  \end{enumerate}
\end{definition}

There are several other ways to characterize the pullback fibrations.
For instance, \eqref{pullback-fibration.2} is equivalent to $\sharp$ having a ``fibered terminal object'' (see \cite{hermida1999fibrations}, Definition~4.7 and Corollary~4.9);
hence the pullback fibrations are exactly the fibrations that have a fibered terminal object and a base category with all finite limits.
Another characterization can be given in terms of adjoints and pullbacks, as in the following fact.

\begin{fact}
  \label{term-fibered}
  Given any functor $\sharp : \EE \to \BB$, it is a pullback fibration iff
  \begin{enumerate}
	  \setcounter{enumi}{\value{equation}}
	\item \label{term-fibered.1} $\sharp$ has a full and faithful
          right adjoint $q:\BB\to\EE$, making $\BB$ a reflective
          subcategory of $\EE$. Moreover, the counit of the adjunction
          is strictly the identity (rather than just an
          isomorphism). We write $\eta : \id_{\EE} \to q \circ \sharp$
          for the unit.
	\item \label{term-fibered.2} $\BB$ has a terminal object $1$.
	\item \label{term-fibered.3} Any arrows of $\EE$ of the form
          $q f : q Y \to q X$ and $h : A \to q X$ have a pullback
	  \[
          \begin{tikzcd}
            B \arrow[r, "\pi_2"] 
            \arrow[d, "\pi_1", swap]
            \arrow[dr, phantom, "\LRcorner", very near start]
            & A \arrow[d, "h"] \\
	    q Y \arrow[r, "q f"] & q X.
          \end{tikzcd}
	  \]
          Moreover, the following diagram is also a pullback, where
          $\tilde h : \sharp A\to X$ is the adjoint mate of $h:A\to
          qX$ and similarly for $\pi_1$. Note that $\tilde h = \sharp
          h$ because $\sharp qX = X$.
	  \[
          \begin{tikzcd}
            \sharp B
            \arrow[dr, phantom, "\LRcorner", very near start]
            \arrow[r, "\sharp \pi_2"]
            \arrow[d, "\tilde \pi_1", swap]
            & \sharp A \arrow[d, "\tilde h"] \\
	    Y \arrow[r, "f"]& X.
          \end{tikzcd}
	  \]
	  In addition, if $X = \sharp A$ and $h = \eta_A : A \to q\sharp
          A$, then $qf$ and $h$ have a pullback $B$ as above such that
          $\sharp B = Y$, $\tilde\pi_1=\id_{Y}$, $\tilde h=\id_{X}$,
          and $\sharp\pi_2=f$ (hence $\pi_2$ is the cartesian lifting
          of $f$).
          \setcounter{equation}{\value{enumi}}
\end{enumerate}
\end{fact}

Note that, when $\sharp$ is a Grothendieck fibration, \eqref{term-fibered.1} is equivalent to $\sharp$ having a fibered terminal object (\cite{jacobs1999categorical}, Lemma~1.8.8).
However, Fact~\ref{term-fibered} does \emph{not} assume that $\sharp$ is a fibration.

\begin{proof}
  Suppose $\sharp :\EE \to \BB$ is a pullback fibration.
  First note that \eqref{pullback-fibration.2} implies \eqref{term-fibered.2} trivially. Next, \eqref{pullback-fibration.2} implies \eqref{term-fibered.1} by Lemma~1.8.8 of \cite{jacobs1999categorical}---%
  but here is a concrete description of the functor $q : \BB \to \EE$ and the adjunction:
  \begin{itemize}
	\item Given any object $X$ of $\BB$, we define $qX$ to be the domain
	  of $\overline{!_X} : qX \to 1$, where $\overline{!_X}$ is the
	  cartesian lifting of $!_X : X \to 1$.
	\item Given any morphism $f : Y \to X$ of $\BB$, we define $qf$ to be the unique arrow below.
	  \[
	  \begin{tikzcd}
		qY \arrow[dr, dashed, "qf"] \arrow[drr, bend left = 20, "!_{qY}"]
		\arrow[d, mapsto, "\sharp"] & & \\
		Y \arrow[dr, "f", swap] & qX \arrow[d, mapsto, "\sharp"] \arrow[r]& 1 \arrow[d, mapsto, "\sharp"]\\
		& X \arrow[r, "!_X", swap] & 1
	  \end{tikzcd}
        \]
      \item Given a map $f : \sharp A \to X$, we define a map
        $g : A \to qX$ by the following diagram.
        \[
          \begin{tikzcd}
            A \arrow[dr, dashed, "g"] \arrow[drr, bend left = 20, "!_{A}"]
            \arrow[d, mapsto, "\sharp"] & & \\
            \sharp A \arrow[dr, "f", swap] & qX \arrow[d, mapsto, "\sharp"] \arrow[r]& 1 \arrow[d, mapsto, "\sharp"]\\
            & X \arrow[r, "!_X", swap] & 1
          \end{tikzcd}
        \]
        On the other hand, given a map $h : A \to qX$, we can obtain
        $\sharp h : \sharp A \to X$ since $\sharp q X = X$. 
  \end{itemize}
  Note that
  $q$ is full and faithful because
  $\BB(A, C) = \BB(\sharp q A, C) = \EE(q A, q C)$.
  To show \eqref{term-fibered.3}, fix any $qf : qY \to qX$ and $h : A \to qX$.
  By the definition of pullback fibrations, the leftmost square below is a pullback in $\BB$.
  We take the cartesian lifting of $\pi_2$ as in the middle, and have the pullback on
  the right by Fact~\ref{pullback-cartesian}.
  \[
	\begin{tikzcd} Z \arrow[r, "\pi_2"] \arrow[d, "\pi_1", swap]
\arrow[dr, phantom, "\LRcorner", very near start] & \sharp A \arrow[d,
"\tilde{h}"] & B \arrow[d, maps to, "\sharp"] \arrow[r, "g"] \arrow[dr, phantom,
"\LRcorner", very near start] & A \arrow[d, maps to, "\sharp"] & B \arrow[r,
"g"] \arrow[dr, phantom, "\LRcorner", very near start] \arrow[d, "\eta_{B}"] & A
\arrow[d, "\eta_{A}"] \\ Y \arrow[r, "f"] & X & Z \arrow[r, "\pi_2"] & \sharp A &
qZ \arrow[r, "q\pi_2"] & q\sharp A
	\end{tikzcd}
  \]
  Apply $q$ to the left pullback, stack it with the right pullback, and we have the following.
  \[
	\begin{tikzcd}
	  B \arrow[r, "g"] \arrow[d, "\eta_B", swap] \arrow[dr, phantom, "\LRcorner", very near start]
	  & A \arrow[d, "\eta_A"]         \\
	  qZ \arrow[r] \arrow[dr, phantom, "\LRcorner", very near start]
	  \arrow[d, "q\pi_1", swap] & q\sharp A \arrow[d, "q\tilde{h}"] \\
	  qY \arrow[r, "qf"]           & qX
	\end{tikzcd}
  \]
  Note that $q\tilde{h} \circ \eta_A = h$, and thus we have the desired pullback.
  The rest of the conditions of \eqref{term-fibered.3} are implied by the above construction.  

  Assuming \eqref{term-fibered.1}--\eqref{term-fibered.3}, we want to prove $\sharp$ is a pullback fibration. First we need to
  show $\sharp$ is a Grothendieck fibration. For any object $A$ of $\EE$,
  suppose we have a morphism $f : X \to \sharp A$ of $\BB$. Then \eqref{term-fibered.3} gives
  the following pullback.
  \[
    \begin{tikzcd}
      B \arrow[r, "\pi"] \arrow[d, "\eta_B", swap] & A \arrow[d, "\eta_A"] \\
      qX \arrow[r, "qf"]  & q\sharp A
    \end{tikzcd}
  \]
  Note that $X = \sharp B$ and $f = \sharp \pi$. By Fact~\ref{pullback-cartesian},
  $\pi$ is the cartesian lifting of $f$. Thus $\sharp$ is a Grothendieck fibration.
  It can be shown that $\BB$ has pullback from \eqref{term-fibered.3} and $\sharp \circ q = \id_{\BB}$.
  The category $\EE$ has a terminal object $q1$, and $\sharp q 1 = 1$. 
\end{proof}

It is useful to reiterate how a pullback fibration $\sharp$ defines its right adjoint $q$ in \eqref{term-fibered.1}:
since $\sharp 1$ is terminal in $\BB$, we set $q \sharp 1 = 1$ and define $qX$ and $qf$ by cartesian liftings.

\begin{notation}
  By abuse of notation, we may write $X$ for $q X$ (since $q$ is a full embedding), and $\und{A}$ for both $\sharp A$ and $q \sharp A$.
\end{notation}

We use the total category $\EE$ of such a fibration $\sharp$ as a category of quantum structures.
It may not have all pullbacks, but the ones given by $\sharp$ are enough for interpreting dependent types.

\subsection{State-parameter fibration}
\label{fib:param-state}

Our semantics takes a pullback fibration $\sharp$ of a symmetric monoidal closed category $\EE$ of quantum (and other) state spaces over a locally cartesian closed subcategory $\BB$ of parameter spaces.
A state space $A$ of $\EE$ then sits over a parameter space $\sharp A$ via $\eta_A : A \to q \sharp A$.
With this we express a parameter-indexed family of quantum states---formally by assuming:

\medskip
\begin{enumerate}
	\setcounter{enumi}{\value{equation}}
  \item
	\label{monoidal-closed.0}
	$\sharp$ is a monoidal closed functor.
	\setcounter{equation}{\value{enumi}}
\end{enumerate}
\medskip
This means that $\sharp$ preserves the monoidal closed structure, i.e., it maps $\otimes$, $\multimap$, $I$ in $\EE$ to $\times$, $\Rightarrow$, $1$ in $\BB$.

In this setup, cartesian arrows are understood as simply reindexing parameters, without any effect on quantum states.
This understanding implies that

\medskip
\begin{enumerate}
	\setcounter{enumi}{\value{equation}}
  \item
	\label{state-parameter.1}
	If arrows $f$ and $g$ of $\EE$ are cartesian, so is $f \otimes g$.
	\setcounter{equation}{\value{enumi}}
\end{enumerate}
\medskip
In short, \eqref{monoidal-closed.0} and \eqref{state-parameter.1} mean that the Grothendieck fibration $\sharp$ is a monoidal fibration \cite{shulman2008fibrations}.

Our setup also gives rise to several useful structures.
One is a ``sublocal'' monoidal structure:

\begin{definition}\label{def:fibered-monoidal-product}
  Let $\sharp : \EE \to \BB$ be a pullback fibration such that $\EE$
  is monoidal closed, $\BB$ is an LCCC, and $\sharp$ preserves the
  monoidal closed structure.  Therefore each slice category $\BB / X$
  is cartesian closed, with $\times_X$ and $\Rightarrow_X$.  Let $A,
  B$ be objects in $\EE / qX$.  We then define the \textit{fibered
    monoidal product} $A \otimes_{qX} B$ and the \textit{fibered
    function space} $A \multimap_{qX} B$ as the following cartesian
  liftings.
  \[
	\begin{tikzcd}
	  A \otimes_{qX} B
	  \arrow[r, tail, "j"]
	  \arrow[d, mapsto, "\sharp"]
	  &
	  A \otimes B
	  \arrow[d, mapsto, "\sharp"]
	  \\
	  \und{A} \times_{X} \und{B}
	  \arrow[r, tail, "i"]
	  &
	  \und{A} \times \und{B}
	\end{tikzcd}
	\hspace{40pt}%
	\begin{tikzcd}
	  A \multimap_{qX} B
	  \arrow[r, tail, "j"]
	  \arrow[d, mapsto, "\sharp"]
	  &
	  A \multimap B
	  \arrow[d, mapsto, "\sharp"]
	  \\
	  \und{A} \Rightarrow_{X} \und{B}
	  \arrow[r, tail, "i"]
	  &
	  \und{A} \Rightarrow \und{B}
	\end{tikzcd}
      \]
      Here the arrows denoted $i$ are canonical inclusion maps, and the arrows denoted $j$ are their cartesian liftings.
\end{definition}

Observe that $A \otimes_{qX} B$ and $A \multimap_{qX} B$ are also objects of $\EE / qX$.
We will assume that

\medskip
\begin{enumerate}
	\setcounter{enumi}{\value{equation}}
  \item
	\label{state-parameter.2}

	$-\otimes_{qX} A \dashv A \multimap_{qX} -$ for every object $X$ of $\BB$,
	\setcounter{equation}{\value{enumi}}
\end{enumerate}
\medskip
so that each $\EE / q X$ is monoidal closed.
We call $\EE$ ``sublocally'' monoidal closed:
while $\EE / A$ is not monoidal closed for every object $A$, it is if $A$ is in the subcategory $\BB$.

Another structure is a functor that gives parameter-in\-dexed families of monoidal units:

\begin{definition}
  \label{parameter-functor}
  Given a pullback fibration $\sharp : \EE \to \BB$ as in Definition~\ref{def:fibered-monoidal-product}, we define a functor $p : \BB \to \EE$ by setting $p 1 = I$ and taking cartesian liftings of the entire $\BB$, or equivalently, pulling back ${!}_{p 1} = \eta_I : I \to \ssunderline{I}$.
  \[
	\begin{tikzcd}
	  p Y
	  \arrow[r, "p f"]
	  \arrow[d, mapsto, "\sharp"]
	  &
	  p X
	  \arrow[r, "p {!}_X"]
	  \arrow[d, mapsto, "\sharp"]
	  &
	  p 1
	  \arrow[d, mapsto, "\sharp"]
	  \\
	  Y \arrow[r, "f"] & X \arrow[r, "{!}_X"] & 1
	\end{tikzcd}
	\hspace{20pt}
	\begin{tikzcd}
	  p Y
	  \arrow[r, "p f"]
	  \arrow[d, "\eta_{pY}", swap]
	  \arrow[dr, phantom, "\LRcorner", very near start]
	  &
	  p X
	  \arrow[r, "p {!}_X"]
	  \arrow[d, "\eta_{pX}", swap]
	  \arrow[dr, phantom, "\LRcorner", very near start]
	  &
	  p 1
	  \arrow[d, "{!}_{p 1} = \eta_I"]
	  \\
	  q Y \arrow[r, "q f"] & q X \arrow[r, "q {!}_X"] & q 1
	\end{tikzcd}
  \]
\end{definition}

It follows that $\sharp \circ p = \id_\BB$.
We see $p X$ as a family of the monoidal unit $I$ indexed by parameters in $X$;
it is indeed the monoidal unit of $\EE / q X$.
(One may also note that $p X = I \times q X$.)
We may therefore refer to each $p X$ as an ``$X$-parameterized unit.''
We will assume that

\medskip
\begin{enumerate}
	\setcounter{enumi}{\value{equation}}
  \item
	\label{state-parameter.3}
	$p$ has a right adjoint $\flat$,
	\setcounter{equation}{\value{enumi}}
\end{enumerate}
\medskip
which is enough to make $p \dashv \flat$ a \emph{linear-nonlinear adjunction} as defined by Benton \cite{benton1994mixed}.
We then write $!$ for the comonad $p \circ \flat$ and $\force : {!} \to \id_\EE$ for the counit.

Putting together Definitions \ref{def:model.pullbacks.cartesian} through \ref{parameter-functor} and conditions \eqref{pullback-fibration.1}--\eqref{state-parameter.3}, we arrive at the following definition.

\begin{definition}
  \label{state-parameter}
  We say that a pullback fibration $\sharp : \EE \to \BB$ over an LCCC $\BB$ (see {\eqref{pullback-fibration.1}}--{\eqref{term-fibered.3}}) is a \emph{state-parameter fibration} if \eqref{monoidal-closed.0}--\eqref{state-parameter.3} hold.
  We call the associated $q, p : \BB \to \EE$ the \emph{parameterized-terminal} and \emph{parameterized-unit} functors, respectively, of $\sharp$.
\end{definition}

\begin{fact}
  Given any state-parameter fibration $\sharp$, its para\-meterized-unit functor $p$ is strong monoidal, i.e., $p(X \times Y) \cong p X \otimes pY$ and $p 1 \cong I$, and $p \dashv \flat$ is a linear-nonlinear adjunction.
\end{fact}
\begin{proof}
  $p 1 \cong I$ by definition.
  To show $p(X \times Y) \cong p X \otimes pY$, note that $p {!}_X : p X \to p 1 \cong I$ and $p {!}_Y$ are cartesian;
  hence $p {!}_X \otimes p {!}_Y : p X \otimes p Y \to p 1 \otimes p 1 \cong p 1$ is also cartesian by \eqref{state-parameter.1}, while $\und{p X \otimes p Y} \cong X \times Y$.
  Thus $p X \otimes pY \cong p(X \times Y)$ by Definition~\ref{parameter-functor}.
  By \cite[Theorem 1.5]{kelly1974doctrinal}, strong monoidal $p$ makes $\flat$ monoidal and $p \dashv \flat$ a linear-nonlinear adjunction.
\end{proof}

\subsection{\texorpdfstring{Example:\@ $\mbbar$ over $\Set$}{Example: M-bar-bar over Set}}
\label{concrete}

Rios and Selinger's \cite{RiosS17} categorical model of the language Proto-Quipper-M naturally gives rise to a state-parameter fibration.
This concrete example may help the reader see how our idea of states and parameters works conceptually.

In the definition of the model, we consider a fixed, but arbitrary
symmetric monoidal category $\m$, and call its morphisms ``generalized
circuits''. Typically, $\m$ is freely generated by some set of objects
called ``wires'' and some set of morphisms called ``gates''; however,
nothing in the following construction relies on any properties of $\m$
other than being symmetric monoidal.  Next, we fully embed $\m$ into
some symmetric monoidal closed category with products $\mbar$ by a
suitable method, e.g., the Yoneda embedding. The model is then defined
as follows.

\begin{definition}
  Let $\mbbar$ be the free coproduct completion of $\mbar$.
  Concretely, this means:
  \begin{itemize}
	\item
	  An object $A$ is a pair $(\und{A}, (A_x)_{x \in \und{A}})$ of a set $\und{A}$ and an indexed family of objects $A_x$ of the category $\mbar$.
	\item
	  An arrow $f : A \to B$ is a pair $(\und{f}, (f_x)_{x \in \und{A}})$ of a function $\und{f} : \und{A} \to \und{B}$ and an indexed family of arrows $f_x : A_x \to B_{\und{f}(x)}$ of $\mbar$.
  \end{itemize}
\end{definition}

Objects and arrows of $\mbbar$ are families of objects and morphisms
of $\mbar$, indexed by parameters in sets. The free coproduct
completion is a well-known construction
\cite{ADAMEK2020106972,borceux1994handbook,jacobs1999categorical}.  The
obvious functor $\sharp : \mbbar \to \Set :: A \mapsto \und{A}$ is a
Grothendieck fibration: the cartesian lifting of a function
$f : X \to \und{A}$ at $A$ has domain $(X, (A_{f(x)})_{x \in X})$ and
consists of $f$ paired with a family $(\id_{A_{f(x)}})_{x \in X}$ of
identity arrows---thus, cartesian arrows simply reindex parameters and
have no effect on states.  Since $\sharp$ preserves the terminal
object (and since $\Set$ is an LCCC), it is a pullback fibration.
While $\sharp$ has a full and faithful right adjoint
$q : \Set \to \mbbar :: X \mapsto (X, (1)_{x \in X})$, there is also
$p : \Set \to \mbbar :: X \mapsto (X, (I)_{x \in X})$, which can be
obtained as in Definition~\ref{parameter-functor}.  Since $p$ has a
right adjoint
$\flat : \mbbar \to \Set :: A \mapsto \sum_{x \in \und{A}} \mbar(I,
A_x)$, property \eqref{state-parameter.3} holds.

The category $\mbbar$ is symmetric monoidal closed, with $I$, $\otimes$, $\multimap$ defined by indexed families of $I$, $\otimes$, $\multimap$ of $\mbar$ with parameter sets that make $\sharp$ monoidal closed:
$I = (1, I)$, $A \otimes B = (\und{A} \times \und{B}, (A_x \otimes B_y)_{x \in \und{A}, y \in \und{B}})$, and $A \multimap B = (\und{A} \Rightarrow \und{B}, (\prod_{x \in \und{A}} (A_x \multimap B_{f(x)}))_{f : \und{A} \to \und{B}})$.
This definition and the characterization of cartesian arrows above imply \eqref{state-parameter.1}.
Moreover, $A \otimes_{q X} B$ and $A \multimap_{q X} B$ are defined to be subfamilies of $A \otimes B$ and $A \multimap B$ (with smaller parameter sets $\und{A} \times_X \und{B}$ and $\und{A} \Rightarrow_X \und{B}$), which implies \eqref{state-parameter.2}, that $\mbbar / q X$ is monoidal closed.
To sum up,

\begin{fact}
  \label{m++:param-state}
  The functor $\sharp : \mbbar \to \Set$ is a state-parameter fibration. 
\end{fact}

In the remainder of this section and Section~\ref{abstract}, we explore how state-parameter fibrations can model linear dependent types.
Fact~\ref{m++:param-state} then means that Proto-Quipper-M can be extended with dependent types, as we will see in Section~\ref{quantum}.

\subsection{Structures for dependent types}
\label{semantics.dependent}

The sublocally monoidal closed structure of a state-parameter fibration $\sharp : \EE \to \BB$ enables us to interpret dependent types with several constructions.

First let us note that any arrow $\pi : A \to q X$, i.e., any object of $\EE / qX$, factors as $\und{\pi} \circ \eta_A$.
We take this to signify that $A$, intuitively an $\und{A}$-indexed family of state spaces, is dependent, via $\und{\pi}$, on parameters in $X$.
In the same vein, any arrow of $\EE / qX$ is parameterized by $X$.
Based on this idea, we use $\EE / qX$ as the category of (objects and arrows interpreting) types and terms dependent on parameters in $X$.
So, schematically, a context $\Gamma \vdash$, a dependent type $\Gamma \vdash A : \ast$, and a term $\Gamma \vdash M : A$ are interpreted respectively by $\Gamma$, $\pi$, and $M$ in

\[
  \begin{tikzpicture}[x=20pt,y=20pt]
	\coordinate (O) at (0,0);
	\foreach \x in {0,1,2} {
	  \node [outer sep=0pt,inner sep=2pt] (v\x) at ({cos((-4 * \x + 5)*pi/6 r)},{sin((-4 * \x + 5)*pi/6 r)}) {};
	};
	\node (A) [inner sep=0.25em] at (v0) {$\Gamma$};
	\node (B) [inner sep=0.25em] at (v1) {$A$};
	\node (X) [inner sep=0.25em] at (v2) {$\und{\Gamma}$};
	\draw [->] (A) -- (B) node [pos=0.5,inner sep=4pt,above] {$M$};
	\draw [->] (A) -- (X) node [pos=0.5,inner sep=4pt,left] {$\eta_\Gamma$};
	\draw [->] (B) -- (X) node [pos=0.5,inner sep=4pt,right] {$\pi$};
  \end{tikzpicture}
\]
In the LCCC interpretation of nonlinear dependent type theory, one would use an identity arrow $\id_\Gamma : \Gamma \to \Gamma$ on the left of the triangle.
In our semantics, it is crucial to use $\eta_\Gamma : \Gamma \to \und{\Gamma}$ instead---because, while $M$ is dependent on parameters, it may use resources contained in its domain.
That is, to generalize the LCCC interpretation to the linear dependent case, we need to split the context object into two:
a parameterized terminal $\und{\Gamma}$ as the base of the slice $\EE / \und{\Gamma}$ to express dependency, and a (generally) linear object $\Gamma$ as the domain of a term to embody linear resources.

One may still note that, when we apply the functor $\sharp$ to the above triangle, $\und{\eta_\Gamma} = \id_{\uline{\Gamma}}$ becomes the usual triangle from the nonlinear case.
In other words, our semantics uses a fibration $\sharp$ of the linear type theory of $\EE$ over the nonlinear dependent type theory of parameters of the LCCC $\BB$.

Now, keeping in mind that $\pi : B \to \und{A}$ means the dependency of $B$ on (parameters of) $A$, we define:

\begin{definition}
  Given $\pi : B \to \und{A}$, i.e., when $B$ is in $\EE / \und{A}$, we call $A \otimes_{\und{A}} B$ a \emph{dependent monoidal product}, and write $A \depotimes B$ for it.
\end{definition}

Dependent monoidal products will be used to interpret linear $\Sigma$-types $(x : A) \otimes B$, but they will also prove necessary in interpreting contexts.
If $\Gamma, x : A \vdash$ is a well-formed context, we must have $\Gamma \vdash A : \ast$, and therefore an arrow $\pi : A \to \und{\Gamma}$; then $\Gamma \depotimes A$ is the interpretation of $\Gamma, x : A \vdash$.
In the nonlinear case, an object $X$ would interpret the context $\Gamma, x : X \vdash$, but in the linear case, we must not drop $\Gamma$ from $\Gamma \depotimes A$, since the linear resource contained in $\Gamma$ is crucial.
Observe however that $\und{A \depotimes B} = \und{A} \times_\und{A} \und{B} \cong \und{B}$.
Thus, as far as parameters are concerned, the context $\Gamma, x : A \vdash$ has a parameter space $\und{A}$, agreeing with the nonlinear dependent type theory of $\BB$.
One may in addition observe that $\depotimes$ is associative.

For linear $\Pi$-types, we introduce the following definition:

\begin{definition}
  Let $A$ be an object in $\EE / q X$ and $B$ an object in $\EE / \und{A}$ (and hence in $\EE / q X$).
  Let $\und{\Pi}_{X, \und{A}} \und{B}$ be the usual dependent function space in $\BB / X$, which is defined by the pullback
  \[
  \begin{tikzcd}
    \und{\Pi}_{X, \und{A}} \und{B}  \arrow[d, "!"] \arrow[r, tail, "i"] & \und{A} \Rightarrow_{X} \und{B}\arrow[d, "A\Rightarrow_{X} \pi"]\\
    1 \arrow[r, "\tilde{\id}"] & \und{A} \Rightarrow_{X} \und{A}.
  \end{tikzcd}
  \]
  We define the dependent function space $\Pi_{qX, A} B$ in $\EE/qX$ as the following cartesian lifting.
  \[
	\begin{tikzcd}
	  \Pi_{qX, A} B \arrow[r, tail, "j"]
	  \arrow[d, mapsto, "\sharp"]
	  & A \multimap_{qX} B \arrow[d, mapsto, "\sharp"]\\
	  \und{\Pi}_{X, \und{A}} \und{B}  \arrow[r, tail, "i"] & \und{A} \Rightarrow_{X} \und{B}
	\end{tikzcd}
      \]
\end{definition}

This can interpret linear $\Pi$-types $(x : A) \multimap B$ due to the following theorem, which provides a semantics for the abstraction and application of linear dependent functions.

\begin{theorem}\label{adj:dep}
  There is an adjunction
  \begin{gather*}
	\begin{tikzpicture}[x=25pt,y=25pt]
	  \node (A) [inner sep=0.25em] at (0,0) {$\EE / \und{A}$};
	  \node (B) [inner sep=0.25em] at (4,0) {$\EE / qX.$};
	  \draw [transform canvas={yshift=-5pt},->] (A) -- (B) node [pos=0.5,inner sep=4pt,below] {$\Pi_{qX, A}$};
	  \draw [transform canvas={yshift=5pt},->] (B) -- (A) node [pos=0.5,inner sep=4pt,above] {${-} \otimes_{qX} A$};
	  \path (A) -- (B) node [pos=0.5] {$\bot$};
	\end{tikzpicture}
  \end{gather*}
  In other words, there is a natural transformation $\epsilon : (\Pi_{qX, A} B) \otimes_{qX} A \to B$. For any morphism $f : C \otimes_{qX} A \to B$, there exists a unique $\widetilde{f} : C \to \Pi_{qX, A} B$ such that $\epsilon \circ (\widetilde{f}\otimes_{qX} A) = f$.
  \[
    \begin{tikzcd}
      C \otimes_{qX} A \arrow[dr, "f"] \arrow[d, "\widetilde{f}\otimes_{qX} A", swap]& \\
      \Pi_{qX, A} B\otimes_{qX} A \arrow[r, "\epsilon", swap]& B
    \end{tikzcd}
  \]
\end{theorem}
\begin{proof}
  In a state-parameter fibration, we have the adjunction $-\otimes_{qX} A \dashv A \multimap_{qX} -$. We write $\varepsilon : (A\multimap_{qX} B) \otimes_{qX} A \to B$ for the counit.  
  For any $f : C \otimes_{qX} A \to B$, we write
  $\widehat{f} : C \to A\multimap_{qX} B$ for the unique morphism such that
  $\varepsilon \circ (\widehat{f} \otimes_{qX} A) = f$.
  Recall that $j : \Pi_{qX, A}B \to A \multimap_{qX}B$ is the cartesian lifting of $\und{j} : \und{\Pi}_{X, \und{ A}}\und{ B} \to \und{A} \Rightarrow_X \und{ B}$.
  We define $\epsilon = \varepsilon \circ (j \otimes_{qX} A) : (\Pi_{qX, A} B) \otimes_{qX} A \to B$.
  Suppose we have $f : C \otimes_{qX} A \to B$, a morphism over $\und{A}$.
  Thus $\und{ f} : \und{ C} \times_X \und{ A} \to \und{ B}$ is a morphism in $\BB$. Since
  $\BB$ is locally cartesian closed, there exists a unique $\widetilde{\und{ f}} : \sharp C \to \und{\Pi}_{X, \und{ A}} \und{ B}$ such that $\sharp (\widehat{f}) = \und{ j} \circ \widetilde{\und{ f}}$. By the cartesianness of $j$, there exists a unique morphism $\widetilde{f} : C \to \Pi_{qX, A}B$ such that $j \circ \widetilde{f} = \widehat{f}$ and $\sharp (\widetilde{f}) = \widetilde{\sharp{ f}}$.
  Thus $\epsilon \circ (\widetilde{f}\otimes_{qX} A) = \varepsilon \circ (j \otimes_{qX} A) \circ (\widetilde{f}\otimes_{qX} A) =  \varepsilon \circ ((j \circ \widetilde{f})\otimes_{qX} A)  = \varepsilon \circ (\widehat{f} \otimes_{qX} A) = f$.
\end{proof}

\subsection{Parameters and shapes}
\label{parameters.shapes}

The linear type theory of $\EE$ is fibered over the nonlinear dependent type theory of parameters in $\BB$, so that dependent types are dependent on parameters.
This role of parameters will be reflected in the syntax of the type theory we will propose in Section~\ref{abstract}, which specifies special subclasses of types and terms, namely, \emph{parameter types} and \emph{parameter terms}.
They refer primarily to objects and arrows in $\BB$.
Nonetheless, given our objective of treating parameter types and other types uniformly, we need to formally interpret parameter types and terms in $\EE$.
We do this by embedding $\BB$ into $\EE$ in two ways, namely, by the functors $q$ and $p$:
while the image of $q$ provides slices $\EE / q X$ and sublocally monoidal closed structure, parameterized monoidal units $p X$ and reindexing maps $p f$ interact with other objects and arrows within the monoidal structure of $\EE$ or $\EE / q X$, in the way that was explained in Section~\ref{semantics.dependent}.

For this reason, we interpret the shape operation on types and terms by $\sharp$ but also by the two endofunctors $q \circ \sharp$ and $p \circ \sharp$ on $\EE$, which we dub the \emph{shape-terminal functor} and \emph{shape-unit functor}, respectively.
While $q \circ \sharp$ is a monad (of $\sharp \dashv q$), $p \circ \sharp$ is not, but one may note that $p \circ \sharp \dashv q \circ \flat$.
The following equalities regarding $p \circ \sharp$ will be useful in interpreting shapes in Section~\ref{abstract}.

\begin{fact}
  \label{shape:semantics}
  The shape-unit functor $p \circ \sharp$ satisfies the following for any $X$ in $\BB$ and any $A$, $B$ in $\EE$:

  \begin{itemize}
	\item $p(\und{pX}) = pX$.
	\item $p \und{S} = I$ for any $S$ in $\EE$ such that $\und{S} = 1$.
	\item $p(\und{A \otimes B}) = p \und{A} \otimes p \und{B}$.
	\item $p(\und{A \otimes_{qX} B}) = p \und{A} \otimes_{qX} p \und{B}$.
	\item $p(\und{!A}) = {!}A$.
	\item $p(\und{A \multimap B}) = p(\und{A} \Rightarrow \und{B})$.
	\item $p(\und{\Pi_{qX, A} B}) = p (\und{\Pi}_{X, \und{A}} \und{B})$.
  \end{itemize}
\end{fact}

Note that $p(\und{A \multimap B}) \neq p\und{A} \multimap p\und{B}$, since $p$ does not map $\Rightarrow$ to $\multimap$.
While $p(\und{A \multimap B}) = p(\und{A} \Rightarrow \und{B})$ is not an exponential in $\EE$, it is (trivially) an exponential in $p\BB$, the (non-full) subcategory of images of $p$.
Parameterized units and cartesian arrows in $p\BB$ correspond to a state-free fragment of $\EE$, in which (nonlinear) function abstraction and application are interpreted by the exponential structure of $\BB$ (or its image under $p$).

\section{A linear dependent type system}
\label{abstract}

In \autoref{model} we provided a general concept of state-parameter fibration.
We now consider a general linear dependent type system that is sound with respect to any given state-parameter fibration $\sharp : \EE \to \BB$.
An extension of this system to quantum programming---%
viz., dependently typed Proto-Quipper---%
will be provided in \autoref{quantum}.

\subsection{Syntax and typing rules}

As was mentioned in Section~\ref{parameters.shapes}, the syntax of our type theory specifies subclasses of types and terms, namely, \emph{parameter types} and \emph{parameter terms}, in parallel to the LCCC $\BB$ being the subcategory of $\EE$ with no linear resources or state.
Furthermore, our syntax allows only parameter terms to appear in types.
This restriction has the consequence that the evaluation of types during type checking will not change quantum states.

\begin{definition}[Syntax]
  \label{syntax}
  \[
    \begin{array}{l}
      \textit{Types}\ 
      A, B ::= C \mid \Unit \mid {!}A  \mid (x : A) \multimap B[x] \mid (x : A) \otimes B[x] \mid (x : P_1) \to P_2[x]
      \\ 
      \textit{Parameter types}\ 
      P ::= \Unit \mid {!}A \mid (x : P_1) \otimes P_2[x] \mid (x : P_1) \to P_2[x]
      \\
      \textit{Terms}\ 
      M, N, L ::= \unitt \mid x \mid \lambda x . M \mid M N \mid \force M \mid \forceprime R \mid \liftt M \mid (M, N)
      \\
      \quad \mid \lett (x, y) = N \tin M \mid \lambda' x . R \mid R_1 @ R_2
      \\ 
      \textit{Parameter terms}\ 
      R ::= \unitt \mid x \mid \lambda' x . R \mid R_1 @ R_2 \mid \forceprime R \mid \liftt M \mid (R_1, R_2)
      \\
      \quad \mid \lett (x, y) = R_1 \tin R_2
      \\
      \textit{Values}\ 
      V ::= \unitt \mid x \mid \lambda x. M \mid \lambda' x. R \mid \liftt M
      \\
      \textit{Indices}\ 
      k ::= 0 \mid 1 \mid \omega
      \\
      \textit{Contexts}\ 
      \Gamma ::= {\cdot} \mid x :_k A, \Gamma
      \text{,\quad where $k=\omega$ only if $A$ is a parameter type.}
      \\
      \textit{Parameter contexts}\ 
      \Phi ::= {\cdot} \mid x :_k A, \Phi
      \text{,\quad where $k\in\s{1,\omega}$ only if $A$ is a parameter type.}
    \end{array}\]
\end{definition}

Here, $C$ ranges over a set of base types for states. In the
application to quantum programming, these might be types of qubits,
qutrits, classical bits, etc. 
We have the usual linear exponential type $!A$, which is a parameter type and is duplicable. 
The term $\liftt M$ introduces the type $!A$, and the term $\force M$ eliminates $!A$ to $A$.
In the linear $\Pi$- and $\Sigma$-types $(x : A) \multimap B[x]$ and $(x : A)\otimes B[x]$, we allow $A$ to be any type.
We write $A \multimap B$ and $A \otimes B$ for them in the special case in which $B$ does not depend on $x : A$.

The linear $\Sigma$-type $(x : A)\otimes B[x]$ is a parameter type when $A$ and $B[x]$ are.
The same is not the case for the $\Pi$-type:
we introduce a separate, intuitionistic $\Pi$-type $(x : P_1) \to P_2[x]$.
This is introduced and eliminated by parameter terms $\lambda' x.R$ and $R_1 @ R_2$, respectively.

The term $\forceprime R$ eliminates $!A$ to the shape of $A$, i.e., $\Sh(A)$.
We interpret $\forceprime$ as the shape of the counit $\force : {!} \to \id_\EE$ of the adjunction $p \dashv \flat$.

The operations $\lambda'$, $@$, and $\forceprime$, as well as the corresponding type former $\to$, are not part of the ``surface language'', i.e., they are not intended to be used directly by the user of the programming language. 
Rather, they are generated by the type checker and only play a role in type checking and during evaluation.

We use the indices $0$, $1$, $\omega$ to keep track of the use of variables in a typing context. Here, $x:_{0}A$ means that the variable $x$ is unused, $x:_{1}A$ means that it is used exactly once, and $x:_{\omega}A$ means that it is used an indeterminate number of times.
In a well-formed context, a linear type can only have the index $0$ or $1$, and not $\omega$. In a parameter context, a linear type can only have the index $0$. 
On the other hand, a parameter type can have any index.
Addition and multiplication of the indices are defined to be commutative so that $0$ and $1$ are the additive and multiplicative units, with $k + \ell = \omega$ for $k, \ell \neq 0$ and with $0 \cdot k = 0$ and $\omega \cdot \omega = \omega$.
It is straightforward to extend these operations pointwise to contexts:
given two contexts $\Gamma_1 = (x_1 :_{k_1} A_1, \ldots, x_n :_{k_n} A_n)$ and $\Gamma_2 = (x_1 :_{\ell_1} A_1, \ldots, x_n :_{\ell_n} A_n)$ with the same sequence of types, we write $\Gamma_1 + \Gamma_2 = (x_1 :_{k_1 + \ell_1} A_1, \ldots, x_n :_{k_n + \ell_n} A_n)$, while $k \Gamma_1 = (x_1 :_{k \cdot k_1} A_1, \ldots, x_n :_{k \cdot k_n} A_n)$.

Parameter contexts $\Phi$ are typing contexts in which all linear types have the index $0$.
Note that parameter contexts are just a special kind of contexts, rather than being formally distinct from other contexts $\Gamma$.

One of the most significant pieces of machinery in our type system is the \textit{shape operation},
which maps a term to a parameter term and a type to a parameter type.
It is interpreted by the shape-unit functor $p \circ \sharp$,
and the type-part of the operation corresponds to equalities in Fact~\ref{shape:semantics}.

\begin{definition}[Shape operation]
  \label{shape}
  \[
	\begin{array}{l}
	  \Sh(C) = \Unit \\
	  \Sh(!A)  = \ !A \\
	  \Sh((x : A) \multimap B[x]) = (x : \Sh(A)) \to \Sh(B[x])\\
	  \Sh((x : A) \otimes B[x]) = (x : \Sh(A)) \otimes \Sh(B[x])\\
	  \Sh(P) = P \\[1ex]
	  \Sh(\unitt) = \unitt \\
	  \Sh(x) = x \\
	  \Sh(\lambda x. N) = \lambda' x. \Sh(N) \\
	  \Sh(M N) = \Sh(M) @ \Sh(N)\\
	  \Sh(\force M) = \forceprime \Sh(M)\\
	  \Sh(\liftt M) = \liftt M\\
	  \Sh(M, N) = (\Sh(M), \Sh(N))\\
	  \Sh(\lett (x, y) = N \tin M) = \lett (x, y) = \Sh(N) \tin \Sh(M)\\
	  \Sh(R) = R
	\end{array}\]
\end{definition}

Observe that the shape operation is idempotent,
just like the shape-unit functor $p \circ \sharp$.
Also note that the shape operation is a meta-operation on terms, like
substitution, and not a term constructor.
The shape operation is intended to be applied to well-typed terms.
Although we define $\Sh(x) = x$, the type of the variable $x$ is changed
according to the shape operation (see Theorem~\ref{shape:syntax}).

Let us write $\Sh(\Gamma)$ to mean applying the shape operation to all the
types in $\Gamma$. So $\Sh(\Gamma)$ is a parameter context. 
We often ignore the index information of a variable if it is of a parameter type
$P$. 

The well-formedness of a context depends on kinding.
The main purpose of a well-formed context is to make sure that a linear type
cannot have index $\omega$. The next three definitions are mutually
recursive, as is common in dependent type theory.

\begin{definition}[Well-formed context]
  \fbox{$\Gamma \vdash$}
  \[
	\begin{array}{llll}
	  \infer{\cdot \vdash}{}
	  &\quad
		\infer[
		\begin{tabular}{@{}r@{}}
		  where $k = \omega$ only if $A$ is a
		  parameter type
        \end{tabular}
		]{\Gamma, x :_k A \vdash }{\Gamma \vdash & \Sh(\Gamma) \vdash A : *}
	\end{array}
  \]
\end{definition}

In the above definition, we use the shape of the context $\Sh(\Gamma)$ to
kind check a type.

\begin{definition}[Kinding rules]
  \label{kinding}
  \fbox{$\Phi \vdash A : *$}
  \[
    \begin{array}{cccc}
      \infer{\Phi \vdash (x : A) \multimap B[x] : *}{\Phi, x : \Sh(A) \vdash B[x] : *}
      &
		\infer{\Phi \vdash (x : A) \otimes B[x] : *}{\Phi, x : \Sh(A) \vdash B[x] : *}
      &
		\infer{\Phi \vdash {!}A : *}{\Phi \vdash A : *}
      &
		\infer{\Phi \vdash (x : P_1) \to P_2[x] : *}{\Phi, x : P_1 \vdash P_2[x] : *}
    \end{array}
  \]
\end{definition}

We only use parameter contexts for kinding, since only parameter terms
can appear in types.
The kinding rules for $(x : A) \otimes B[x]$ and $(x : A) \multimap B[x]$ suggest
that the type $B$ only depends on the shape of $A$, i.e., $\Sh(A)$, which
is a parameter type.

\begin{definition}[Typing rules] \fbox{$\Gamma \vdash M : A$}
  \label{typing}
  \[ \small
    \begin{array}{cc}
      \\
      \infer{\Gamma_1 + \Gamma_2 \vdash M N :
B[\Sh(N)]}{\Gamma_1 \vdash M : (x : A) \multimap B[x] & \Gamma_2
\vdash N : A} & \infer{0\Gamma, x :_1 A, 0\Gamma' \vdash x :
A}{0\Gamma, x :_1 A, 0\Gamma' \vdash } \\ \\ \infer{\Gamma \vdash
\lambda x . M : (x : A) \multimap B[x]}{\Gamma, x :_k A \vdash M :
B[x] & k \neq 1 \Rightarrow A = P}

      &

        \infer{\Phi \vdash \liftt M :\ ! A}{\Phi \vdash M : A}

      \\ \\

      \infer{\Gamma_1 + \Gamma_2 \vdash (M, N) : (x : A) \otimes
B[x]}{\Gamma_1 \vdash M : A & \Gamma_2 \vdash N : B[\Sh(M)]}

      &
        
        \infer{\Gamma \vdash \force M : A}{\Gamma \vdash M :\ !A}

      \\ \\ \infer {\Gamma_1+ \Gamma_2 \vdash \lett (x, y) = M \tin N
: C} {
      \begin{array}{l} \Gamma_1 \vdash M : (x : A) \otimes B[x] \\
\Gamma_2, x :_{k_1} A, y :_{k_2} B[x] \vdash N : C\\ k_1 = 0
\Rightarrow A = P_1, k_2 = 0 \Rightarrow B[x] = P_2[x]
      \end{array}} & \infer{\Phi \vdash \forceprime R : \Sh(A)}{\Phi
\vdash R :\ !A} \\ \\ \infer{\Phi \vdash R_1 @ R_2 : P_2[R_2]}{\Phi
\vdash R_1 : (x : P_1) \to P_2[x] & \Phi \vdash R_2 : P_1} &
\infer{\Phi \vdash \lambda' x . R : (x : P_1) \to P_2[x]} {\Phi, x :
P_1 \vdash R : P_2[x]}
    \end{array}
  \]
\end{definition}

\begin{remark*}
  (i) In the typing rule for $M N$, the parameter term $\Sh(N)$ is substituted into the type,
  therefore the term $N$ only occurs once in the term $M N$ and its shape $\Sh(N)$
  can occur many times in the type $B[\Sh(N)]$.

  (ii) In the typing rule for $\lambda x. M$, we do not allow abstracting over a variable of linear type with index $0$, because this violates linearity. This differs from the formulation used by McBride and by Atkey, where abstracting over zero uses of a resource is allowed.

  (iii) In the typing
  rule for $\liftt M$, we
  require the context to be a parameter context $\Phi$. Although the term $M$ may not
  be a parameter term, the term $\liftt M$ is considered a parameter term
  because it has an interpretation in $p\BB$.

  (iv) In the typing rule for $(M, N)$, the term
  $M$ also appears as a parameter term $\Sh(M)$ in the type $B$. All linear variables
  of $M$ will have index $0$ in the context $\Gamma_2$ and index $1$ in $\Gamma_1$,
  so $M$ is considered used exactly once in the term $(M, N)$.

  (v) We identify a subset of typing rules $\Phi \vdash R : P$ for typing parameter terms.
  We do not formally separate the typing rules into two fragments. In the semantics,
  the interpretation of $\Phi \vdash R : P$ is an arrow in $p\BB$.

  (vi) The typing rules for the parameter terms $R_1 @ R_2$ and $\lambda' x. R$ are the usual ones from intuitionistic dependent types,
  where the contexts are required to be parameter contexts. The
  parameter term $\forceprime R$ has a type of the form $\Sh(A)$.

  (vii) We assume all contexts appearing in typing rules to
  be well-formed (i.e., if it is not well-formed, the rule does not
  apply). In particular, this applies to contexts of the form
  $\Gamma_1+\Gamma_2$, which carry the side condition that a type with
  index $\omega$ must be a parameter type.
\end{remark*}

We have the following standard syntactic results: a well-formed typing judgment
implies that the type and the context are well-formed, and the shape operation preserves typing. 

\begin{theorem}\label{shape:syntax}
  \
  \begin{itemize}
	\item
	  If $\Gamma \vdash M : A$, then $\Sh(\Gamma) \vdash A : *$ and $\Gamma \vdash$.
	\item If $\Gamma \vdash M : A$, then $\Sh(\Gamma) \vdash \Sh(M) : \Sh(A)$.
  \end{itemize}

\end{theorem}
\begin{proof}
  By induction on the derivation of $\Gamma \vdash M : A$. 
\end{proof}
Since parameter terms have no side effects (i.e., do not affect states), we can freely substitute them into
a term or a type. We cannot freely substitute a general term because
side effects may get duplicated after the substitution. 
Since a value can be of a linear type, when substituting it into a type, we
must apply the shape operation to the value. To summarize, we have:

\begin{theorem}[Substitution]\label{substitution}
  \
  \begin{itemize}
	\item
	  If $\Phi, x : P, \Phi' \vdash B[x] : \ast$ and $\Phi \vdash R : P$,
	  then $\Phi, [R/x]\Phi' \vdash B[R] : \ast$.
	\item If $\Phi, x : P, \Gamma \vdash M : B[x]$ and $\Phi \vdash R : P$, then $\Phi, [R/x]\Gamma \vdash [R/x]M : B[R]$.
	\item If $\Gamma_1, x :_k A, \Gamma' \vdash M : B[x]$
	  and $\Gamma_2 \vdash V : A$, then $\Gamma_1 + k \Gamma_2, [\Sh(V)/x]\Gamma' \vdash
	  [V/x]M : B[\Sh(V)]$.
  \end{itemize}
  
\end{theorem}
\begin{proof}
  By induction on the kinding and typing derivations.  
\end{proof}

\subsection{Operational semantics}

We adopt a big-step call-by-value operational semantics for our type system.
A term in our language may in principle modify the state.
However, since we have not yet introduced any primitive operations that can actually update the state, the state has been omitted from the statement of the following evaluation rules. 
We will return to the state change in Section~\ref{op:dpq}, when we discuss a concrete language with specific primitive operations.

\begin{definition}[Evaluation rules]  \fbox{$M \Downarrow N$}
  \label{evaluation}
  \[
    \small
	\begin{array}{l@{\qquad}l@{\qquad}l}
      \infer{M N \Downarrow N'}{M \Downarrow \lambda x.M' & N \Downarrow V & [V/x]M' \Downarrow N'}
      &
        \infer{\force M \Downarrow N}{M \Downarrow \liftt M' & M' \Downarrow N}
          &
		\infer[\textit{force}']{\forceprime R  \Downarrow R'}
		{R \Downarrow  \liftt M & \Sh(M) \Downarrow R'}
                                          
        \end{array}
      \]
      \[
        \begin{array}{l@{\qquad}l}
                \infer{R_1 @ R_2  \Downarrow R'}
	      {R_1 \Downarrow \lambda' x. R_1' &
	        R_2 \Downarrow V & [V / x ]R_1' \Downarrow R'}
      &
          \infer{\lett (x, y) = N \tin M \Downarrow  N'}{ N \Downarrow (V_1, V_2) &  [V_2/y]([V_1/x]M) \Downarrow N'}
        \end{array}        
      \]
\end{definition}

Note that in the $\textit{force}'$ rule, a term of the form $\liftt M$ is
a parameter term, even when $M$ is not. So to ensure
that the resulting parameter term $\forceprime (\liftt M)$ does not
modify state, we must evaluate it to $\Sh(M)$, which is again
a parameter term.

Since
the evaluation rules are also used during
type checking, the evaluation is defined on open terms.
Thus it may not always evaluate a term to a value. 
We treat variables as values to facilitate type-level evaluation.

We have the
following type conversion rule for equality of types.

\begin{definition}[Type conversion]
  \[
    \begin{array}{l}
	  \infer{\Gamma \vdash M : A[R']}
	  {
	  \Gamma \vdash M : A[R] &
							   R \Downarrow R'
							   }
    \end{array}
  \]
\end{definition}

The type preservation for the big-step evaluation can be proved by induction on the
evaluation rules.
The main insight that is needed to prove type preservation is that only
values and parameter terms can be substituted into another term.

\begin{proposition}[Type preservation]
  \label{preservation}
  If $\Gamma \vdash M : A$ and $ M \Downarrow M'$,
  then we have $\Gamma \vdash M' : A$. 
\end{proposition}

\subsection{Denotational semantics}
\label{interpretation}

We interpret our type system with a state-parameter fibration, by providing contexts $\Gamma \vdash$, dependent types $\Gamma \vdash A : \ast$, and terms $\Gamma \vdash M : A$ with interpretations $\interp{-}$ in the manner described in Section~\ref{semantics.dependent}, and by interpreting the shape operation by the shape-unit functor $p \circ \sharp$ as in Section~\ref{parameters.shapes}.
In addition, the substitution of a term into a type corresponds to a pullback in the standard fashion.
This semantics is encapsulated in the following theorems (and their proofs).
Recall that $\und{A}$ is a notation for $\sharp A$. We often write $\interp{\Gamma}$ instead of $\interp{\Gamma\vdash}$.

\begin{theorem}
  \label{interp}
  An interpretation $\interp{-}$ can be defined in such a way that
  \begin{enumerate}
    \setcounter{enumi}{\value{equation}}
  \item\label{interp.1}
    $\interp{\Gamma}$ is an object of\/ $\EE$.
  \item\label{interp.2}
    $\interp{\Phi \vdash A : \ast}$ is an object of\/ $\EE/\und{\interp{\Phi}}$.
    We may write $\pi : \interp{A} \to \und{\interp{\Phi}}$ for it.
  \item\label{interp.3}
    $\interp{\Gamma \vdash M : A}$ is an arrow of\/ $\EE/\und{\interp{\Gamma}}$ from $\eta_{\interp{\Gamma}}$ to $\interp{\Gamma \vdash A : \ast}$.

  \item\label{interp.4} 
    $\interp{\Sh(\Gamma)} = p\und{\interp{\Gamma}}$, an object of\/ $p\BB$.
  \item\label{interp.5}
    $\interp{\Phi \vdash \Sh(A) : \ast} = p \und{\interp{\Phi \vdash A : \ast}}$,
    an object of\/ $p\BB/\und{\interp{\Phi}}$.
  \item\label{interp.6}
    $\interp{\Sh(\Gamma) \vdash \Sh(M) : \Sh(A)} = p\und{\interp{\Gamma \vdash M : A}}$,
    an arrow of\/ $p\BB/\und{\interp{\Gamma}}$.

  \item\label{interp.7}
    Suppose $\pi = \interp{\Phi, x : P \vdash B : \ast} : \interp{B} \to \und{\interp{\Phi, x : P}}$
    and $\interp{R} = \interp{\Phi \vdash R : P} : \interp{\Phi} \to \interp{P}$ in $p\BB/\und{\interp{\Phi}}$.
    Then $\pi' = \interp{\Phi \vdash [R/x]B : \ast} : \interp{[R/x]B} \to \und{\interp{\Phi}}$ is the pullback of\/ $\pi$ along $\und{\interp{R}}$ as in
    \[
      \begin{tikzcd}
        \interp{[R/x]B} \arrow[r]
        \arrow[dr, phantom, "\LRcorner", very near start]
        \arrow[d, "\pi'", swap]
        & \interp{B} \arrow[d, "\pi"] \\
        \und{\interp{\Phi}} \arrow[r, "\und{\interp{R}}"] &
        \und{\interp{\Phi, x : P}} = \und{\interp{P}} 
      \end{tikzcd}
    \]
    \setcounter{equation}{\value{enumi}}
  \end{enumerate}
\end{theorem}

\begin{proof}[Proof sketch] The proof of this theorem goes by simultaneous induction on the derivation, with each case of the inductive step showing how exactly the semantics works.
Here we show some examples of the cases (for a more detailed proof see Appendix~\ref{app:proofs}).

\noindent\smallskip
(i)
The first is
\begin{gather*}
  \infer[
  \begin{tabular}{@{}l@{}}
    $k = \omega$ only if $A$ is
    a parameter type
  \end{tabular}
  ]{\Gamma, x :_k A \vdash }{\Gamma \vdash & \Sh(\Gamma) \vdash A : *}
\end{gather*}
By the assumptions, we have objects $\interp{\Gamma}$ of $\EE$ and $\pi : \interp{A} \to \und{\interp{\Sh(\Gamma)}} = \und{\interp{\Gamma}}$ of $\EE / \und{\interp{\Gamma}}$.
We therefore define
\begin{gather*}
  \interp{\Gamma, x :_k A} = \interp{\Gamma}\depotimes \interp{A}^k
\end{gather*}
using a dependent monoidal product, where $\interp{A}^k = p \und{\interp{A}}$ if $k = 0$ and $\interp{A}^k = \interp{A}$ otherwise.

\noindent\smallskip
(ii)
The second case is
\begin{gather*}
  \infer{\Phi \vdash (x : A) \multimap B[x] : *}{\Phi, x :_k \Sh(A) \vdash B[x] : *}
\end{gather*}
The assumption implies $\Phi \vdash A : \ast$.
Hence we have arrows $\pi: \interp{B} \to \und{\interp{\Phi, x :_k \Sh(A)}} = \und{\interp{A}}$ and
$\interp{A} \to \und{\interp{\Phi}}$.
We therefore define
\begin{gather*}
  \interp{\Phi \vdash (x : A) \multimap B[x] : *} = \Pi_{\und{\interp{\Phi}}, \interp{A}} \interp{B} \to \und{\interp{\Phi}},
\end{gather*}
using a dependent function space that gives an object in $\EE/\und{\interp{\Phi}}$.

\noindent\smallskip
(iii)
The third case is 
\begin{gather*}
  \infer{0\Gamma, x :_1 A, 0\Gamma' \vdash x : A}{0\Gamma, x :_1 A,  0\Gamma' \vdash }      
\end{gather*}

Since $\Sh(\Gamma) \vdash A : *$, by (11) we have $\pi : \interp{A} \to \und{\interp{\Gamma}}$. Note that there is a canonical morphism $f : p\und{\interp{\Gamma}}\depotimes \interp{A}\depotimes  p\und{\interp{\Gamma'}} \to \interp{A}$ over $\und{\interp{\Gamma}}$. We define $\interp{x}$ as the following
$u$ over $\und{\interp{\Gamma, x : A, \Gamma'}}$.
\[
  \begin{tikzcd}
    p\und{\interp{\Gamma}}\depotimes \interp{A}\depotimes  p\und{\interp{\Gamma'}} \arrow[d, dashed, "u"] \arrow[dr, "f"] \arrow[dd, "\eta", bend right=60, swap]&  \\
    \und{\interp{\Gamma'}}\times_{\und{\interp{\Gamma}}} \interp{A} \arrow[r] \arrow[d, "\pi'"] \arrow[dr, phantom, "\LRcorner", very near start]& \interp{A} \arrow[d, "\pi"]  \\
    \und{\interp{\Gamma, x : A, \Gamma'}} = \und{\interp{\Gamma'}} \arrow[r] &
    \und{\interp{\Gamma}}
  \end{tikzcd}
\]

\noindent\smallskip
(iv)
The fourth case is
\begin{gather*}
  \infer{\Gamma \vdash \lambda x . M : (x : A) \multimap B[x]}
  {\Gamma, x :_k A \vdash M : B[x] & k = 0 \Rightarrow A = P}
\end{gather*}

\noindent By the assumptions, we have a morphism $\interp{M} : \interp{\Gamma}\otimes_{\und{\interp{\Gamma}}} \interp{A}^k \to \interp{B}$ over $\und{\interp{\Gamma, x : A}}$. By the adjunction in Theorem~\ref{adj:dep},
we define $\interp{\lambda x. M} = \widetilde{\interp{M}} : \interp{\Gamma} \to \prod_{\und{\interp{\Gamma}}, \interp{A}^k} \interp{B}$ as a morphism over $\und{\interp{\Gamma}}$.

\noindent\smallskip
(v)
The last case we review here is
\begin{gather*}
  \infer{\Gamma_1 + \Gamma_2 \vdash M N : B[\Sh(N)]}
  {\Gamma_1 \vdash M : (x : A) \multimap B[x] &
    \Gamma_2 \vdash N : A}
\end{gather*}
\noindent The assumptions give arrows
$\interp{M} : \interp{\Gamma_1} \to \Pi_{\und{\interp{\Gamma}}, \interp{A}} \interp{B}$ and
$\interp{N} : \interp{\Gamma_2} \to \interp{A}$
of $\EE/\und{\interp{\Gamma}}$ for $\und{\interp{\Gamma}} = \und{\interp{\Gamma_1}} = \und{\interp{\Gamma_2}}$,
as well as $\pi: \interp{B} \to \und{\interp{\Sh(\Gamma), x :_k \Sh(A)}} = \und{\interp{A}}$.
Note that $\und{\interp{\Sh(N)}} = \und{\interp{N}}$, since \eqref{interp.6} implies $\interp{\Sh(N)} = p \und{\interp{N}}$.
Hence \eqref{interp.7} gives the pullback square in the following.
We therefore define $\interp{M N}$ as the unique arrow $u$ that makes the diagram commute.
\begin{gather*}
  \begin{tikzpicture}[x=20pt,y=20pt]
    \coordinate (O) at (0,0);
    \coordinate (r) at (6,0);
    \coordinate (d) at (0,-2);
    \node (A1) [inner sep=0.25em] at (O) {$\interp{\Gamma_1 + \Gamma_2}$};
    \node (A2) [inner sep=0.25em] at ($ (A1) + (r) $) {$(\Pi_{\und{\interp{\Gamma}}, \interp{A}} \interp{B}) \otimes_{\und{\interp{\Gamma}}} \interp{A}$};
    \node (B1) [inner sep=0.25em] at ($ (A1) + 0.6*(r) + 0.6*(d) $) {$\interp{B[\Sh(N)]}$};
    \node (B2) [inner sep=0.25em] at ($ (B1) + (r) $) {$\interp{B}$};
    \node (C1) [inner sep=0.25em] at ($ (B1) + (d) $) {$\interp{\Gamma}$};
    \node (C2) [inner sep=0.25em] at ($ (C1) + (r) $) {$\und{\interp{A}}$};
    \draw [->] (B1) -- (C1) node [scale=0.7,pos=0.5,inner sep=2pt,right] {$\pi$};
    \draw [->] (B2) -- (C2);
    \draw [->] (C1) -- (C2) node [scale=0.7,pos=0.5,inner sep=2pt,below] {$\und{\interp{N}}$};
    \draw [->] (B1) -- (B2);
    \draw [->] (A1) -- (A2) node [scale=0.7,pos=0.5,inner sep=2pt,above] {$\interp{M} \otimes_{\und{\interp{\Gamma}}} \interp{N}$};
    \draw [->] (A2) -- (B2) node [scale=0.7,pos=0.75,anchor=240,inner sep=4pt] {$\epsilon$};
    \path (A1) edge [->,bend right=15] coordinate [pos=0.5] (A1-C1-label) (C1);
    \node [scale=0.7,anchor=north east,inner sep=1pt] at (A1-C1-label) {$\eta$};
    \draw [->,dotted] (A1) -- (B1) node [scale=0.7,pos=0.5,anchor=60,inner sep=2pt] {$u$};
    \coordinate (B1-pb) at ($ (B1) + (1.4,-0.7) $);
    \draw ($ (B1-pb) + (-0.15,0) $) -- (B1-pb) -- ($ (B1-pb) + (0,0.15) $);
  \end{tikzpicture}
\end{gather*}
It is crucial to observe that $\interp{\Gamma_1+\Gamma_2} = \interp{\Gamma_1}\depotimes \interp{\Gamma_2}$.
This then enables us to use $\epsilon : (\Pi_{\und{\interp{\Gamma}}, \interp{A}} \interp{B}) \otimes_{\und{\interp{\Gamma}}} \interp{A}\to \interp{B}$,
the counit of the adjunction of Theorem~\ref{adj:dep}. 
\end{proof}

For this semantics, we have the following soundness theorem.
The main step in proving this theorem is using the semantic version of
the substitution theorem (Theorem~\ref{substitution}).

\begin{theorem}[Soundness]
  \label{soundness}
  If $\Gamma \vdash M : A$ and $ M \Downarrow M'$,
  then we have $\interp{M} = \interp{M'}$. 
\end{theorem}
\begin{proof}
  By induction on the derivation of $M \Downarrow M'$.
  We only consider the following case.
  \[
      \infer{ M N  \Downarrow  N'}{ M \Downarrow \lambda x. M' &
         N \Downarrow V &  [V/x]M' \Downarrow N'}
  \]
 We perform  a case analysis on $\Gamma \vdash M N : A$.
  \[
        \infer{\Gamma_1 + \Gamma_2  \vdash M N : A[\Sh(N)]}{\Gamma_1 \vdash M : (x : B) \multimap A[x] & \Gamma_2 \vdash N : B}
  \]
    Note that $\Gamma = \Gamma_1 + \Gamma_2$.
    So we have $\Gamma_1 \vdash M : (x : B) \multimap A[x]$.
    By Proposition~\ref{preservation}, we have $\Gamma_1 \vdash  \lambda x.M' : (x : B) \multimap A[x]$
    and $\Gamma_2 \vdash V : B$. By inversion, we have $\Gamma_1, x :_1 B \vdash M' :  A[x]$ (here we are supposing $k=1$). By syntactic substitution
    (Theorem~\ref{substitution}),
    we have $\Gamma \vdash [V/x]M' : A[\Sh(V)]$.
    So by type equality, we have $\Gamma \vdash [V/x]M' : A[\Sh(N)]$.
    Thus we have $\Gamma \vdash [V/x]M' : A[\Sh(N)]$
    and  $\Gamma \vdash N' : A[\Sh(N)]$.

    On the semantics side, we have $\interp{M} = \interp{\lambda x.M'}$, $\interp{N} = \interp{V}$ and $\interp{[V/x]M'} = \interp{N'}$. It is sufficient to show $\interp{M N} = \interp{[V/x]M'}$.
    Consider the following diagram.
  \[
      \begin{tikzcd}
        \interp{\Gamma_1+\Gamma_2} \arrow[r, "\interp{\Gamma_1} \depotimes \interp{V}"]
        \arrow[d, dashed, "\interp{M N}"]
        \arrow[d, dashed, "\interp{[V/x]M'}", bend right = 30, swap]
        \arrow[dr, "\epsilon \circ (\widetilde{\interp{M'}} \depotimes \interp{V})"]
        &
        \interp{\Gamma_1, x:_k B} \arrow[d, "\interp{M'}"]\\
        \interp{A[\Sh(N)]} = \interp{A[\Sh(V)]} \arrow[r] \arrow[d]
        \arrow[dr, phantom, "\LRcorner\qquad", very near start] &  \interp{A} \arrow[d]
        \\
        \und{\interp{\Gamma_1 + \Gamma_2}}
        \arrow[r, "\und{\interp{V}} = \und{\interp{N}}"]
        & \und{\interp{\Gamma_1, x: B}}
      \end{tikzcd}
  \]
    To show $\interp{M N} = \interp{[V/x]M'}$, we just need to
    show $\epsilon \circ (\widetilde{\interp{M'}} \depotimes \interp{V}) =\interp{M'}\circ (\interp{\Gamma_1} \depotimes \interp{V})$, which
    holds by properties of the adjunction.
\end{proof}

\section{Dependently typed Proto-Quipper}
\label{quantum}
The type system we defined in Section~\ref{abstract} is an abstract language derived directly from the general categorical structure of state-parameter fibrations.
It is therefore not equipped with concrete features of a particular instance of a fibration---%
such as quantum data types, quantum circuits, and circuit boxing and unboxing.

In this section we show how to support programming quantum circuits. On the
semantic side, this means that we will work with a concrete
state-parameter fibration, i.e., $\sharp : \mbbar \to \Set$. On the syntactic side,
this means that we will extend the type system of Section~\ref{abstract}
with several constructs to work with quantum circuits. We call the resulting
type system ``dependently typed Proto-Quipper'', or sometimes Proto-Quipper-D.

\subsection{Simple types and quantum circuits}
\label{simple-types}
Recall from Section~\ref{concrete} that the model $\mbbar$ of \cite{RiosS17} contains a (full) monoidal subcategory $\m$ of (generalized) circuits.
Objects $S$ of $\m$ (or $(1, S)$ of $\mbbar$) are called \textit{simple}.
It follows that $S$ is simple iff $\sharp S = 1$.
Such objects represent states without parameters (i.e., a single state space, rather than an indexed family of them).
We call a type simple if it is interpreted by a simple object.
For any simple objects $S_1$ and $S_2$, we have the isomorphism
\begin{gather*}
  !(S_1 \multimap S_2) \stackrel{\boxt/ \unboxt}{\cong} p \m(S_1, S_2).
\end{gather*}
This isomorphism means that the function type $!(S_1\multimap S_2)$ literally represents a hom-set of the category $\m$, which we can think of as a set of circuits with input $S_1$ and output $S_2$. 
We equip the programming language with a type $\Circ(S_1, S_2)$ representing such circuits as first-class citizens.
These linear functions of simple types correspond to quantum circuits.

The type $\Qubit$ is defined as $(1, \QQ)$, where
$\QQ$ is a designated object in $\m$ representing the qubit type.
As a result of being a free coproduct completion, $\mbbar$ has coproducts. Thus
Proto-Quipper-M admits quantum data types.
For example,
the type $\List \Qubit$ can be interpreted as an object
$\Sigma_{n \in \NN}(1, \QQ^{\otimes n}) = (\NN, (\QQ^{\otimes n})_{n \in \NN})$.

\paragraph{\textbf{A problem of quantum data types}} The $\boxt/\unboxt$ isomorphism only holds for simple types.
However, $\List \Qubit$ is not a simple type since $\sharp\interp{\List \Qubit} = \sharp (\NN, (\QQ^{\otimes n})_{n \in \NN}) = \NN$. Thus, there is no such type as $\Circ(\List \Qubit, \List\Qubit)$ in Proto-Quipper-M, and a function of type
$!(\List \Qubit \multimap \List \Qubit)$ cannot be converted into a circuit. 
This makes sense, because such a function actually represents a \emph{family} of circuits, indexed by the length of the input list. In practice, however, we want programmers to be able to work with 
quantum data types and linear functions between them. The above
limitation makes boxing circuits in Proto-Quipper-M awkward.

\paragraph{\textbf{Solution via dependent quantum data types}} We solve this problem by working with \textit{dependent quantum data types}, such as $\Vect{\Qubit} n$.
For a value $n$ of type $\Nat$, we define $\interp{\Vect{\Qubit} n}$ as $V_n = (1, \QQ^{\otimes n})$ as in the following pullback square. Note that the
existence of the pullback is guaranteed by the state-parameter fibration $\sharp : \mbbar \to \Set$.
\[
  \begin{tikzcd}
	\interp{\Vect{\Qubit} n} = V_n \arrow[r]
	\arrow[d]
	\arrow[dr, phantom, "\LRcorner", very near start]
	& \interp{\List \Qubit} = \Sigma_{n \in \NN}(1, \QQ^{\otimes n}) \arrow[d]\\
	q 1  \arrow[r, "q n"] & q \NN
  \end{tikzcd}
\]

\noindent Thus $\sharp\interp{\Vect{\Qubit} n} = \sharp V_n = 1$
and $\Vect{\Qubit} n$ is a simple type.
Now a linear function of the type $!(\Vect{\Qubit} n \multimap \Vect{\Qubit} n)$ can be boxed into a circuit, which has the type $\Circ(\Vect{\Qubit} n, \Vect{\Qubit} n)$.
Together with the dependent function type, the programmer can
now define a function of
the type \[!((n : \Nat) \multimap \Circ(\Vect{\Qubit} n, \Vect{\Qubit} n)).\]
Such a function represents a family of circuits indexed by $n$.

We extend the syntax of Definition~\ref{syntax} with simple types,
circuit types, dependent data types, and the boxing and unboxing
operations for circuits.

\begin{definition}[Extended syntax]
  \[
    \small
    \begin{array}{l}
      \textit{Types} \ A, B \ ::= \cdots \mid \Nat\mid \List A \mid \Vect{A} R 
      \mid \Circ(S_1, S_2)
      \\
      \textit{Simple Types}\ S \ ::= \Qubit \mid \Bit \mid \Unit\mid S_1 \otimes S_2 \mid \Vect{S} R
      \\
      \textit{Parameter types}\ P ::= \cdots \mid \Nat \mid \List P \mid \Vect{P} R \mid \Circ(S_1, S_2)
      \\
      \textit{Terms}\ M ::= \cdots \mid \ell \mid (a,\CC,b) \mid \boxt_S M \mid \mathsf{apply}(M, N) \mid \mathsf{apply}'(R_{1}, R_{2})
    \end{array}
  \]
\end{definition}

Note that the vector data type $\Vect{A} R$ requires the length
index to be a parameter term $R$, as only parameter terms can appear
in types.  The type $\Circ(S_1, S_2)$ is considered a parameter type
because
$\interp{\Circ(S_1, S_2)} = p \m(\interp{S_1}, \interp{S_2})$, which
is a parameter object.

Following Proto-Quipper-M, we extend terms with labels $\ell$,
which are distinct from variables. Labels are used to represent
wires of a quantum circuit (or more generally, components of a
tensor product). Unlike variables, labels cannot be substituted. We
also extend contexts with label contexts of the form
$\Sigma = \ell_1 :_{k_1} S_{1},..., \ell_n :_{k_n} S_{n}$, where
$k_i \in \s{0,1}$ and $S_{i} = \Qubit | \Bit$. The semantics of $\Sigma$ is a tensor product of
qubits or bits, and we identify each $\Sigma$ with the corresponding object
in $\m$.

The canonical inhabitants of $\Circ(S_1,S_2)$ are \emph{boxed
  circuits} of the form $(a,\CC,b)$, where $\CC$ is a morphism of
$\m$ and $a,b$ are interfaces connecting the inputs and outputs of
$\CC$ to the types $S_1$ and $S_2$, respectively (see
{\cite{RiosS17}} for more details).

The following are the typing rules for $\Circ(S_1, S_2)$:
\begin{gather*}
  \begin{tabular}{c}
	\infer{ \Phi\vdash (a, \CC, b) :  \Circ(S_1, S_2)}{\Sigma_1 \vdash a : S_1 & \Sigma_2 \vdash b : S_2 & \CC : \Sigma_1 \to \Sigma_2}
  \end{tabular}
  \\
  \begin{tabular}{ccc}
	\infer{\Gamma \vdash \boxt_{S_1} M : \Circ(S_1,S_2)}{\Gamma \vdash M :\ ! (S_1 \multimap S_2)}
	&\quad&
   \infer{\Gamma_{1} + \Gamma_{2} \vdash \mathsf{apply}(M, N) : S_{2}}{\Gamma_{1} \vdash M : \Circ(S_1,S_2) & \Gamma_{2} \vdash N : S_{1}}
  \end{tabular}
  \\
  \begin{tabular}{c}
	\infer{\Phi \vdash \mathsf{apply}'(R_{1}, R_{2}) : \Sh(S_{2})}{\Phi \vdash R_{1} : \Circ(S_1,S_2) & \Phi \vdash R_{2} : \Sh(S_{1})}
  \end{tabular}
\end{gather*}
Boxed circuits are not part of the surface language; rather they are
built and consumed by the $\boxt$ and $\unboxt$ operations. Note that $\unboxt = \lambda c. \liftt \lambda s . \mathsf{apply}(c, s)$.
Certain
built-in boxed circuits, called \emph{gates}, may be bound to
constant symbols of the language (or provided by a standard
library). An implementation may also provide additional primitive
operations on boxed circuits, such as
$\reverse : \Circ(S_1,S_2)\to\Circ(S_2,S_1)$. Indeed, although the
types $!(S_1\multimap S_2)$ and $\Circ(S_1,S_2)$ are semantically
isomorphic, they are operationally distinct, because
$\Circ(S_1,S_2)$ has canonical inhabitants whereas
$!(S_1\multimap S_2)$ does not.

The shape operation can be extended naturally.

\begin{definition}[Extended shape operation]
  \[
    \begin{array}{l}
      \Sh (\Nat) = \Nat
      \\
      \Sh (\List A) = \List \Sh(A)
      \\
      \Sh (\Vect{A} R) = \Vect{\Sh(A)} R
      \\
      \Sh (\Circ(S_1, S_2)) = \Circ(S_1, S_2)
      \\[1ex]
      \Sh(\ell) = \unitt
      \\
      \Sh(a, \CC, b) = (a, \CC, b)
      \\
      \Sh(\boxt_S M) = \boxt_S \Sh(M)
      \\
      \Sh(\mathsf{apply}(M, N)) = \mathsf{apply}'(\Sh(M), \Sh(N)) 
    \end{array}
  \]
\end{definition}
This definition is semantically sound for $\List$ and $\Vect{}$ because
the shape-unit functor $p\sharp$ preserves
tensor products and coproducts in $\mbbar$.
We also observe that $\Sh (\List \Qubit) = \List \Unit \cong \Nat$ and $\Sh (\Vect{\Qubit} R) = \Vect{\Unit} R \cong \Unit$.
\paragraph{\textbf{Safe list-to-vector conversion}} Allowing types to depend on the \textit{shape} of any other types allows us to define a function to \textit{safely} convert
a list of qubits into
a vector of qubits.
\[
  \begin{array}{l}
    \term{conv} : {!}((x : \List \Qubit) \multimap \Vect{\Qubit}\, (\term{toNat}\, x))
	\\
	\term{conv} =\\
	\quad \liftt\,(\lambda x .\, \term{case}\, x \, \term{of} \\
	\quad \quad\quad \quad\quad \Nil \, \to \VNil \\
	\quad \quad\quad \quad\quad \Cons y ys \, \to \VCons y\, (\force\term{conv}\, ys))
  \end{array}
\]
Let us consider the type \[(x : \List \Qubit) \multimap \Vect{\Qubit}\, (\term{toNat}\, x).\] Here $\term{toNat} : \List \Unit \to \Nat$ is the isomorphism that converts a list of units to a natural number. Using the kinding rules of Definition~\ref{kinding} to kind check this type, we only need to kind check the following
\[x : \Sh(\List \Qubit) = \List \Unit \vdash \Vect{\Qubit} \, (\term{toNat}\, x) : *,\]
which is a valid kinding judgment. So a function of the type \[(x : \List \Qubit) \multimap \Vect{\Qubit}\, (\term{toNat}\, x)\]
takes a list of qubits as input
and outputs a vector of qubits, where
the length of the vector is the length of the input list.

The above conversion would not be possible if we required dependent types to be only of the form $(x : P) \to B[x]$, where $P$ is a parameter type. McBride's conversion function has a different
flavor. In \cite[Section 5]{mcbride2016got}, he defines
$\term{conv} :_{\omega} (x :_1 \List X) \multimap \Vect{X} (\length x)$, where $\length :_0 (x :_0 \List X) \multimap \Nat$.
We cannot define a length function of type $\List \Qubit \multimap \Nat$ because it violates linearity.

\subsection{Operational semantics}
\label{op:dpq}
In Proto-Quipper-M \cite{RiosS17}, the call-by-value big-step evaluation is defined on
a \textit{configuration} of the form $(\CC, M)$, where
$\CC$ denotes a morphism of $\m$, and $M$ is a term
that can append quantum gates to $\CC$ when evaluated.

\begin{definition}
  We define a \textit{well-typed configuration}
  \[\Gamma; \Sigma \vdash (\CC, M) : A ; \Sigma'\]
  to mean there exists $\Sigma''$ such that $\CC : \Sigma \to \Sigma''\otimes \Sigma'$ and $\Gamma, \Sigma'' \vdash M : A$.
\end{definition}

Our definition of well-typed configurations allows $M$ to be an open term, with free variables from the
context $\Gamma$. The reason for this is that the evaluation rules can
be used during type checking, where there can be free variables in types.

We can visualize the configuration
$(\CC, M)$ in the above definition as the following diagram,
where the term $M$ in $(\CC, M)$ is using the output wires from $\Sigma''$.

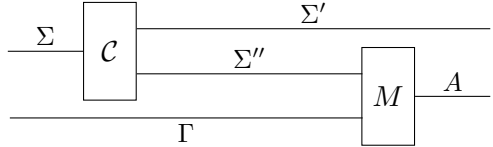
\[
  \begin{tikzpicture}[scale=1,
    box/.style={rectangle,draw,
	  inner sep=1mm,minimum width=7mm,minimum height=12mm}
	]
    \node[box,minimum height=13mm] (f) at (-0.3,0) {$\CC$};
    \node[box,minimum height=13mm] (m) at (3.4,-.6) {${M}$};    
    \draw[-] (f.west) -- node[above, inner sep=2pt]{\small ${\Sigma}$} +(-1cm,0);

    \draw[-] ([yshift=-.3cm]f.east) -- node[above, inner sep=2pt]{\small ${\Sigma''}$} +(3.0cm,0);
    
    \draw[-] ([yshift=.3cm]f.east) -- node[above, inner sep=2pt]{\small ${\Sigma'}$} +(4.7cm,0);

    \draw[-] ([yshift=-.23cm,xshift=-1.33cm]f.south) -- node[below, inner sep=2pt]{\small ${\Gamma}$} ([yshift=-.23cm,xshift=3.36cm]f.south);

    \draw[-] (m.east) -- node[above, inner sep=2pt]{\small ${A}$} +(1cm,0);
  \end{tikzpicture}
\]

Dependently typed Proto-Quipper has the same evaluation rules as
Proto-Quipper-M, plus
rules for reducing parameter terms.

\begin{definition}[Evaluation rules]
  \[
    \infer{(\CC_1, \mathsf{apply}(M, N))  \Downarrow (\CC_4, b)}
    {
      (\CC_1, M) \Downarrow (\CC_2, (a, \DD, b)) &
      (\CC_2, N) \Downarrow (\CC_3, a') &
      \CC_4 = \append(\CC_3, a', (a, \DD, b))
    }
  \]
  \tinyskip
  \[
    \infer{(\CC, \boxt_{S[R]} M) \Downarrow (\CC',  (a,\DD,b )) }
    {
      \begin{array}{@{}c@{}}
        \begin{array}{@{}ccc@{}}
          (\CC, M)\Downarrow (\CC', \liftt \ M') &
                                                   (\id_{I}, R) \Downarrow (\_, V) &
                                                                                     (\id_a, M'\ a) \Downarrow (\DD, b)
        \end{array}
        \\
        \begin{array}{@{}cc@{}}
          a = \gen(S[V]) &
                           \FV(S[V]) = \emptyset
        \end{array}
      \end{array}
    }
  \]
  \tinyskip
  \[
    \infer{(\CC, M N)  \Downarrow (\CC''', N')}
    {
      (\CC, M) \Downarrow (\CC', \lambda x. M') &
      (\CC', N) \Downarrow (\CC'', V) &
      (\CC'', [V/x]M') \Downarrow (\CC''', N')
    }
  \]
  \tinyskip
  \[
    \infer{(\CC, \force M)  \Downarrow (\CC'',N)}
    {(\CC, M) \Downarrow (\CC', \liftt M') &
      (\CC', M') \Downarrow (\CC'', N)}
  \]
  \tinyskip
  \[
    \infer{(\CC_1, R_{1} @R_2) \Downarrow (\CC_4, V')}
    {
      (\CC_1, R_1) \Downarrow (\CC_2, \lambda' x. R_1') &
      (\CC_2, R_2) \Downarrow (\CC_3, V) &
      (\CC_3, [V/x]R_1') \Downarrow (\CC_4, V')
    }
  \]
  \tinyskip
  \[
    \infer{(\CC_1, \forceprime R)  \Downarrow (\CC_3,R')}
    {(\CC_1, R) \Downarrow (\CC_2, \liftt M) & (\CC_1, \Sh(M)) \Downarrow (\CC_3, R')}
  \]
  \tinyskip
  \[
    \infer{(\CC_1, \mathsf{apply}'(R_{1}, R_{2}))  \Downarrow (\CC_3, \Sh(b))}
    {
      (\CC_1, R_{1}) \Downarrow (\CC_2, (a, \DD, b)) &
      (\CC_2, R_{2}) \Downarrow (\CC_3, R_{2}')
    }
  \]
\end{definition}  

In the first evaluation rule, the term $\mathsf{apply}(M, N)$ appends a circuit $\DD$ to the circuit $\CC_3$, resulting in an updated circuit $\CC_4$.
The term $N$ evaluates to a label tuple $a'$, pointing to some of the outputs of the then-current circuit $\CC_3$.
This tuple $a'$ must match the input interface $a$ of the boxed circuit $(a,\DD,b)$.
Finally, the operation $\append(\CC_3, a', (a, \DD, b))$ constructs a new circuit
$\CC_4$ by connecting the outputs $a'$ of $\CC_3$ to the inputs $a$ of $(a, \DD, b)$.
As a result, $\CC_4$ will expose $b$ as part of its output interface.

The second evaluation rule is for boxing a circuit. The notation
$S[R]$ denotes a simple type $S$ that may contain a parameter term
$R$.  The type annotation $S[R]$ on the keyword $\boxt$ is used to
generate a data structure $a = \gen(S[V])$ holding fresh wire
labels. The evaluator then evaluates $M'\ a$ under the identity
circuit $\mathrm{id}_a$, which produces a circuit $\mathcal{D}$
and a circuit output $b$.

The last three
rules are for evaluating parameter terms.
Since parameter terms do not change states, the evaluation
of these terms does not append any quantum gates. We have the following
theorem.
\begin{theorem}
  If $(\CC, R) \Downarrow (\CC', R')$,
  then we have $\CC = \CC'$.
\end{theorem}

The semantics of well-typed configurations is given below. It corresponds to
the above diagram for $(\CC, M)$, where vertical composition is interpreted
as the fibered monoidal product $\otimes_{\und{\interp{\Gamma}}}$~.
\begin{definition}

  Let $\Gamma; \Sigma \vdash (\CC, M) : A ; \Sigma'$ be a well-typed configuration.
  We define
  \[
    \interp{(\CC, M)} = (\interp{M} \depotimes \interp{\Sigma'}) \circ (\interp{\Gamma}\depotimes \CC) :
    \interp{\Gamma} \depotimes \interp{\Sigma}\to \interp{A} \depotimes \interp{\Sigma'} ,
  \]
  a morphism in $\mbbar/\und{\interp{\Gamma}}$.
\end{definition}

Dependently typed Proto-Quipper satisfies type preservation and soundness.

\begin{theorem}[Type preservation and soundness]
  If $\Gamma; \Sigma \vdash (\CC, M) : A ; \Sigma'$ and $(\CC, M) \Downarrow (\CC', M')$,
  then we have $\Gamma; \Sigma \vdash (\CC', M') : A ; \Sigma'$ and
  $\interp{(\CC, M)} = \interp{(\CC', M')}$. 
\end{theorem}

The soundness theorem above, together with the fact that our language
is terminating, implies the following adequacy property:
\begin{theorem}
  If $\emptyset \vdash (\CC, M) : \mathbf{Circ}(S, U); \emptyset$ and
  $\interp{(\CC, M)} = \interp{(\CC', V)}$, then $(\CC, M) \Downarrow
  (\CC', V)$ (up to a renaming of labels).
\end{theorem}

\subsection{Prototype implementation}
We have made a prototype implementation of dependently typed
Proto-Quipper. It is implemented in Haskell as a stand-alone
interpreter and includes a parser, a type checker and an
evaluator. The interpreter provides several options such as allowing
the user to print out quantum circuits and obtain gate count
information. For more information about programming in
Proto-Quipper-D, see the tutorial \cite{FKRS2020}.

We will now discuss some essential features
that enable the type system in this paper to scale
to a usable programming language. 

\paragraph{\textbf{Type inference and elaboration}}
It is a well-known issue that programming with linear types can be quite \textit{invasive}: the
programmer must supply linearity annotations, which is burdensome for large
programs.
Recent research tries to address this problem. For example, Paykin and
Zdan\-ce\-wic {\cite{paykin2017linearity}} propose embedding a linear-nonlinear type system
in Haskell.
Bernardy et al.\@ {\cite{bernardy2017linear}} propose extending Haskell
with linear types.

In our implementation, every top-level declaration must be of
parameter type (since top-level functions are reusable). Programs can
be written without using annotations such as $\liftt$ and $\force$.
We implemented a type elaborator that extends the well-known bi-directional type
inference algorithm \cite{pierce2000local} with the ability to insert linearity annotations. The basic
idea is that when \textit{checking} a program $M$ with type $!A$, the elaborator
produces a new term $\lift M'$, where $M'$ is the result of checking $M$
against the type $A$. When \textit{inferring} a type for the application $M N$,
where $M$ has the type $!(A \multimap B)$,
the elaborator will produce a new term
$(\force M)\, N'$ and a type $B$, where $N'$ is the result of
\textit{checking} $N$ against $A$.
We also implemented a secondary checker to recheck the annotated terms produced by
the elaborator. We have found that this two-step elaboration-and-checking
works well in practice and linear annotations are manageable with
the help of the type elaborator.

\paragraph{\textbf{Simple data types and parameter data types}}
We implement Haskell 98 style \cite{jones2003haskell} data types (which begin with
the keyword \texttt{data}). For example, list data type can be declared as
\texttt{data List a = Nil | Cons a (List a)}.
We also implemented a special kind
of dependent data types that we call simple data type declarations (which begin
with the keyword \texttt{simple}). For example, the following is a declaration
of a vector data type.

\begin{verbatim}
  simple Vec a : Nat -> Type where
     Vec a Z = VNil
     Vec a (S n) = VCons a (Vec a n)
\end{verbatim}

As we mentioned in Section~\ref{simple-types}, only linear functions of simple types can
be boxed into circuits and simple types can be characterized semantically by $\sharp S = 1$.
Simple types can also be described syntactically as types that uniquely determine the size
of their inhabitants. For example, $\List$ is not a simple type because $\List\Qubit$ may
have inhabitants of different sizes, while inhabitants of $\Vect{\Qubit}{3}$ must have size 3, so
it is a simple type.

In Proto-Quipper-D, the simple data type declaration is checked by the compiler
using a syntactic characterization of simple types. Note that such a simplicity check
is undecidable in general, as it is equivalent to checking termination of a recursive function
(e.g., the vector data type declaration above can be viewed as a primitive recursive function which recurs on a second input of
type \texttt{Nat}). We implemented a simplicity checker that only accepts certain
declarations that correspond to primitive recursive functions. 

Since we allow user-defined types, the notion of parameter types and simple types must account for such extensions. 
We implement parameter types and simple types as type classes \cite{wadler1989make}, and their instances are automatically generated by the compiler upon defining a data type. 
For example, the list data type declaration will generate an instance \texttt{(Parameter a) => Parameter (List a)}, which means that \texttt{List a} is a parameter type if
\texttt{a} is a parameter type. Similarly, the vector data type declaration will generate
an instance \texttt{(Simple a) => Simple (Vec a n)}, which means \texttt{Vec a n} is a simple
type if the type \texttt{a} is simple.

With simple type constraints, we
can give the circuit boxing operator $\boxt$ the following type
\texttt{(Simple a, Simple b) => !(a -> b) -> Circ(a, b)}. This way the type checker
is able to check at compile time whether the functions we are boxing are indeed
linear functions of simple types. With parameter type constraints, we
can safely discard reusable variables without violating linearity. For example,
we can define a function \texttt{fst} to retrieve the left element of a pair, which
will have type \texttt{!((Parameter b) => a * b -> a)}. Here \texttt{a * b} means $a \otimes b$. Since \texttt{b} is a parameter type, it is safe
to discard the right element. 

\paragraph{\textbf{Irrelevant quantification}}
Besides the linear dependent types described in this paper, we
also implemented a version of Miquel's irrelevant quantification~\cite{miquel2001implicit,barras2008implicit}. Although the usual dependent quantification is enough in theory,
it is beneficial to implement irrelevant quantification. Because irrelevant arguments
are erased during evaluation, length-indexed
data types such as $\Vect{}$ will not store their length at runtime, hence making the
evaluator more efficient.

In Proto-Quipper-D, a dependent quantification takes the form \texttt{(n :\ Nat) -> T}, and
an irrelevant quantification takes the form \texttt{forall (n :\ Nat) -> T}. We implement
the following typing rules for irrelevant quantification.
\[
  \begin{array}{ccc}
	\infer{\Phi, \Gamma \vdash \lambda \{x\}. M : \forall (x : P). B[x]}{\Phi, x : P, \Gamma \vdash  M : B[x]}
         &\quad&
	\infer{\Phi, \Gamma \vdash M \{R\}:  B[R]}{\Phi, \Gamma \vdash M : \forall (x : P). B[x]
	& \Phi \vdash R : A}
  \end{array}
\]
Both irrelevant abstractions $\lambda \{x\}. M$ and applications
$M \{R\}$ are erased to $|M|$ during evaluation. So we must be
able to discard the irrelevant argument
$R$. Hence we require the irrelevant quantification to be of the form $\forall (x : P). B[x]$ (where $P$ must be a parameter type),
and the irrelevant arguments to be parameter terms.
The shape operation can be extended for irrelevant
quantification.
\begin{align*}
  \Sh(\forall (x : P). B[x]) & = \forall (x : P) . \Sh(B[x]) \\
  \Sh(\lambda \{x\}. M) & = \lambda \{x\}. \Sh(M) \\
  \Sh( M\{R\}) & = \Sh(M) \{R\}
\end{align*}

Other than annotating types using the keyword \texttt{forall}, programmers do not need
to work with irrelevant abstractions and applications explicitly. Our type elaborator
automatically generates the irrelevance annotations, just like linearity annotations.
For example, in Proto-Quipper-D, the vector
append function has type \texttt{!(forall a (n m :\ Nat) -> Vec a n -> Vec a m -> Vec a (add n m))}, but its definition is the same as that of the list append function.

\section{Conclusion and future work}

In this paper, we introduced state-parameter fibrations,
a new categorical structure that extends linear-nonlinear adjunctions for (non-dependent) linear type theory \cite{benton1994mixed} and also generalizes locally cartesian closed categories for (nonlinear) dependent type theory \cite{seely1984locally}.
The ``sublocally'' monoidal closed structure of our fibrations provides a categorical semantics for a general notion of linear dependent types.
We demonstrated this by deriving a linear dependent type system from these fibrations which is sound with respect to the semantics.
We also observed that Rios and Selinger's \cite{RiosS17} model $\mbbar$ of Proto-Quipper-M, their programming language for quantum circuits,
constitutes an instance of a state-parameter fibration $\sharp : \mbbar \to \Set$.
Based on this observation, we extended the general type theory of state-parameter fibrations to a dependently typed extension of Proto-Quipper-M.
The extension, called dependently typed Proto-Quipper or Proto-Quipper-D, is sound with respect to the fibration $\sharp : \mbbar \to \Set$.
We provided a big-step operational semantics for Proto-Quipper-D and proved type preservation and soundness.
Furthermore, we laid out a prototype implementation of Proto-Quipper-D and discussed practical aspects of it.

Let us close the paper by describing several lines of ongoing and future research.
One direction of further investigation is how to expand Proto-Quipper-D with additional features, such as dynamic lifting or recursion.
Note that the version of Proto-Quipper-D described in this paper supports the operation of quantum measurement 
only insofar as it can be viewed as a gate in a circuit. The related
operation of taking the outcome of a measurement (i.e., a state) and
using it to direct the control flow of the programming language (i.e.,
using it as a parameter) is what we call \emph{dynamic lifting}
\cite{green2013quipper}. While the present paper
does not consider dynamic lifting, we gave a categorical semantics of
dynamic lifting (without dependent types) in more recent work. See
\cite{fu2022types,fu2022model} and references therein.
How to integrate dependent types with dynamic lifting is an open question that should be investigated next.
Another open question is how to combine our notion of linear dependent types with recursion.

Also of interest is the relationship between this paper's fibrational
semantics and other styles of semantics of type theory.  In the case
without linearity, it is well-known that fibrations and indexed
categories are essentially equivalent. It is therefore a natural
question to ask whether there is a notion of linear indexed categories
corresponding to our state-parameter fibrations. Another question
relates to the categorical properties of our semantics. In this paper,
we only proved the soundness of our interpretation. One may ask about
completeness, and in particular, whether the type theory's syntactic
category forms a state-parameter fibration (or a slight modification
thereof), completing an analogue to the picture delineated by Seely
\cite{seely1984locally}.

\section*{Acknowledgments}

We thank Neil J. Ross for helpful discussions of a draft of this
paper, Marco Gaboardi and Aaron Stump for pointing out
related works, and Phil Scott for insightful comments on our
categorical semantics.
Our thanks also go to the anonymous referees for LICS 2020 and
this issue of \emph{Logical Methods in Computer Science} for
various suggestions.

This work was supported by the
Air Force Office of Scientific Research under award number
FA9550-15-1-0331. Any opinions, findings and conclusions or
recommendations expressed in this material are those of the authors
and do not necessarily reflect the views of the U.S. Department of
Defense.

\appendix

\section{Some syntactic properties}
\label{sec:some-synt-prop}

The following are a few straightforward syntactic properties.

\begin{lemma}
  \label{props}
  \begin{enumerate}

  \item
    If $\Phi \vdash \Sh(A) : *$, then $\Phi \vdash A : *$ and $\Phi \vdash$.
  \item
    If $\Phi \vdash A : *$, then $\Phi \vdash \Sh(A) : *$ and $\Phi \vdash$. 
    
  \item  
    If $\Gamma \vdash M : A$, then $\Sh(\Gamma) \vdash A : *$ and $\Gamma \vdash$.
    
  \item If $\Gamma \vdash V : P$, then $\Gamma = \Phi$ and $\Sh(V) = V$. 
  \end{enumerate}
\end{lemma}

The following lemma states that the well-formed contexts are stable under
substitution. Note that we only substitute parameter terms for variables in types. 
\begin{lemma}
  \label{sub:context}
  \begin{enumerate}   
  \item If $\Phi, x :_k P, \Gamma \vdash$ and $\Phi \vdash R : P$, then $\Phi, [R/x]\Gamma \vdash$.
  \item If $\Gamma_1, x :_k A, \Gamma' \vdash$ and $\Gamma_2 \vdash V : A$, then $\Gamma_1 + k\Gamma_2 , [\Sh(V)/x]\Gamma' \vdash$.
  \end{enumerate}
\end{lemma}

\begin{theorem}
  \label{sub:syn}
  If $\Phi, x:_k P, \Phi' \vdash A[x] : *$
  and $\Phi \vdash R : P$, then $\Phi, [R/x]\Phi' \vdash A[R]:*$.
\end{theorem}

The following substitution theorem states that we can substitute a parameter term $R$
for a variable $x$ of parameter type in a term $M$. 
\begin{theorem}
  \label{syn:sub1}
  If $\Phi, x :_k P, \Gamma \vdash M : B[x]$ and $\Phi \vdash R : P$, then $\Phi, [R/x]\Gamma \vdash [R/x]M : B[R]$. 
\end{theorem}

The proof of the following theorem is also by induction on derivations.
\begin{theorem}
  \label{syn:sub3}
  If $\Gamma_1, x :_k A, \Gamma' \vdash M : B[x]$
  and $\Gamma_2 \vdash V : A$, then $\Gamma_1 + k \Gamma_2, [\Sh(V)/x]\Gamma' \vdash
  [V/x]M : B[\Sh(V)]$.
\end{theorem}

\begin{theorem}[Type preservation]
  If $\Gamma\vdash M : A$ and $M \Downarrow M'$,
  then we have $\Gamma \vdash M' : A$.
\end{theorem}

\section{Proofs for Section~\ref{interpretation}}
\label{app:proofs}

The following theorem shows that the parameterized-unit functor $p$ maps $\times_Y$ to $\otimes_{qY}$.

\begin{theorem}
  \label{p:product}
  Let $p$ be the parameterized-unit functor. We have $p(X \times_Y Z) = p X \otimes_{qY} p Z$. 
\end{theorem}
\begin{proof}
  We have the following pullback diagram.
  \[
    \begin{tikzcd}
      pX \otimes_{qY} pZ
      \arrow[dr, phantom, "\LRcorner\quad", very near start]
      \arrow[d]
      \arrow[r]
      & p X \otimes p Z = p (X \times Z)
      \arrow[dr, phantom, "\LRcorner", very near start]
      \arrow[r]
      \arrow[d]
      & I \arrow[d]\\
      q(X \times_Y Z) \arrow[r]& q(X \times Z)
      \arrow[r]
      & q 1
    \end{tikzcd}
    \qedhere
  \]
\end{proof}

In the statement of the following lemma, it is intentional that two
different functors $q$ and $p$ are applied to $Z$ and $X$,
respectively. The point is that the category $\EE$ does not in general
have products, but it does have them when one of the components is in
the image of $q$. The lemma shows that if the other component is a
parameter object, then so is the product. This will be used in the
proof of Theorem~\ref{app:shape:interp} below.

\begin{lemma}
  \label{psharp:product}
  Let $\pi_X : X \to Y$ and $\pi_Z : Z \to Y$ be two maps in $\BB$.
  We have $qZ \times_{qY} pX = pZ \otimes_{qY} pX = p\sharp(qZ \times_{qY} pX)$. 
\end{lemma}
\begin{proof}
  Note that $pZ \otimes_{qY} pX = p\sharp(qZ \times_{qY} pX)$ is by Fact~\ref{term-fibered} and
  Theorem~\ref{p:product}. 
  To show $qZ \times_{qY} pX = pZ \otimes_{qY} pX$, we use the following pullback squares.
  \[
    \begin{tabular}{ll}
    \begin{tikzcd}
        pZ \otimes_{qY} pX  \arrow[r]
       \arrow[dr, phantom, "\LRcorner\quad", very near start]
        \arrow[d]
        & pZ \otimes pX \arrow[d] \arrow[r]
        \arrow[dr, phantom, "\LRcorner", very near start]
        & pX \arrow[d] \\
        q(Z \times_Y X) \arrow[r] & q(Z \times X) \arrow[r] & qX
    \end{tikzcd}
      \\
      \\
    \begin{tikzcd}
      qZ \times_{qY} pX  \arrow[rr]
      \arrow[dr, "\eta"]
       \arrow[dr, phantom, "\LRcorner\quad", very near start]
        \arrow[dd]
        && pX \arrow[dr, "\eta"] \arrow[dd] && \\
        & q(Z \times_Y X) \arrow[rr] \arrow[dl]
        \arrow[dr, phantom, "\LRcorner", very near start]
        && qX \arrow[dl]\\
        qZ \arrow[rr] & {} & qY &&
    \end{tikzcd}
    \end{tabular}
  \]
\end{proof}

\begin{lemma}
  \label{stack:pullback}
  Suppose we have the following pullback squares.
  \[
    \begin{tabular}{lll}
    \begin{tikzcd}
        A  \arrow[r]
       \arrow[dr, phantom, "\LRcorner", very near start]
        \arrow[d]
        & B \arrow[d] 
         \\
        qX \arrow[r] & q Y
    \end{tikzcd}
    &
      \begin{tikzcd}
        C  \arrow[r]
       \arrow[dr, phantom, "\LRcorner", very near start]
        \arrow[d]
        & D \arrow[d] 
         \\
        \und{A} \arrow[r] & \und{B}
    \end{tikzcd}
    \end{tabular}
  \]
  Then the following is a pullback square.
  \[
      \begin{tikzcd}
        A\depotimes C  \arrow[r]
       \arrow[dr, phantom, "\LRcorner", very near start]
        \arrow[d]
        & B \depotimes D \arrow[d] 
         \\
        qX \arrow[r] & qY
    \end{tikzcd}
  \]
\end{lemma}

\begin{lemma}
  \label{bq}
  We have $\sharp \circ p \dashv \flat \circ q$ and $\flat \circ q \cong \id_{\BB}$. 
\end{lemma}

\begin{theorem}
  \label{app:interp}
  \
  
  \begin{enumerate}
	\item $\interp{\Gamma \vdash}$ is defined as an object in $\EE$.

	\item $\interp{\Phi \vdash A : *}$ is defined as an object $\pi : \interp{A} \to \und{\interp{\Phi}}$ in the slice category $\EE/\und{\interp{\Phi}}$.
	\item $\interp{\Gamma \vdash M : A}$ is defined as a morphism $\interp{M} : \interp{\Gamma} \to \interp{A}$ in the slice category $\EE/\und{\interp{\Gamma}}$.

  \end{enumerate}
\end{theorem}
\begin{proof}
  \
  
  \begin{enumerate}
	\item By induction on $\Gamma \vdash$.

	  \begin{itemize}
		\item We define $\interp{\cdot \vdash} = I$.
		\item Case:
		  \begin{gather*}
			\infer[
			\begin{tabular}{@{}r@{}}
			  where $k = \omega$ only if $A$ is a
			  parameter type
			\end{tabular}
			]{\Gamma, x :_k A \vdash }{\Gamma \vdash & \Sh(\Gamma) \vdash A : *}
		  \end{gather*}

		  By (2), we have $\pi : \interp{A} \to \und{\interp{\Gamma}}$. We
		  define
		  \[\interp{\Gamma, x :_k A} = \interp{\Gamma}\depotimes \interp{A}^k\]
		  as a fibered tensor product over $\und{\interp{\Gamma}}$. Here $\interp{A}^k = \interp{A}$
		  if $k = 1$, otherwise $\interp{A}^k = p\und{\interp{A}}$.
	  \end{itemize}
	\item By induction on the derivation of $\Phi \vdash A : *$.
	  \begin{itemize}
		\item

		  We define $\interp{C} = A$, where $A$ is a designated object in
		  $\EE$ such that $\sharp A = 1$.

		\item Case:
		  \begin{gather*}
			\infer{\Phi \vdash {!}A : *}{\Phi \vdash A : *}
		  \end{gather*}

		  By the induction hypothesis, we have a morphism $\pi: \interp{A} \to \und{\interp{\Phi}}$.
		  We define \[\interp{\Phi \vdash {!}A} =  q\flat \pi \circ \eta_{p\flat\interp{A}}  : p\flat\interp{A} \to \und{\interp{\Phi}}.\]
		  Note that we have $\flat \circ q \cong \id_{\mathbf{B}}$.

		\item Case:
		  \begin{gather*}
			\infer{\Phi \vdash (x : A) \multimap B[x] : *}{\Phi, x :_k \Sh(A) \vdash B[x] : *}
		  \end{gather*}

		  By the induction hypothesis, we have morphisms $\pi: \interp{B} \to \und{\interp{A}}$ and
		  $\interp{A} \to \und{\interp{\Phi}}$ (since $\Phi \vdash A : *$) in $\EE$.
		  We define \[\interp{\Phi \vdash (x : A) \multimap B[x] : *} = \prod_{\und{\interp{\Phi}}, \interp{A}} \interp{B} \to \und{\interp{\Phi}}\]
		  as an object in $\EE/\und{\interp{\Phi}}$.

		\item Case:
		  \begin{gather*}
			\infer{\Phi \vdash (x : A) \otimes B[x] : *}{\Phi, x :_k \Sh(A) \vdash B[x] : *}
		  \end{gather*}

		  By the induction hypothesis, we have morphisms $\pi: \interp{B} \to \und{\interp{A}}$ and
		  $\interp{A} \to \und{\interp{\Phi}}$ (since $\Phi \vdash A : *$) in $\EE$.
		  We define
		  \[\interp{\Phi \vdash (x : A) \otimes B[x] : *} = \interp{A}\otimes_{\und{\interp{A}}} \interp{B}\]
		  as an object in $\EE/\und{\interp{\Phi}}$.

		\item Case:
		  \begin{gather*}
			\infer{\Phi \vdash (x : P_1) \to P_2[x] : *}{\Phi, x :_k P_1 \vdash P_2[x] : *}
		  \end{gather*}

		  By induction, we have $a_{2} : \interp{P_{2}} \to \und{\interp{P_{1}}}$
		  and $a_{1} : \interp{P_{1}} \to \und{\interp{\Phi}}$.
		  Hence there are maps $\und{\interp{P_{1}}} \to \und{\interp{\Phi}}$
		  and $\und{\interp{P_{2}}} \to \und{\interp{P_{1}}}$.
		  Hence there exists a map
		  $a : \prod_{\und{\interp{\Phi}}, \und{\interp{P_{1}}}} \und{\interp{P_{2}}} \to \und{\interp{\Phi}}$.
		  We define
		  \[ \interp{\Phi \vdash (x : P_1) \to P_2[x] : *} = \eta_{p\und{\interp{\Phi}}} \circ p a :  p(\und{\prod}_{\und{\interp{\Phi}}, \und{\interp{P_1}}} \und{\interp{P_2}}) \to \und{\interp{\Phi}}.\]
	  \end{itemize}

	\item By induction on the derivation of $\Gamma \vdash M : A$.
	  \begin{itemize}

		\item Case:
		  \begin{gather*}
			\infer{0\Gamma, x :_1 A, 0\Gamma' \vdash x : A}{0\Gamma, x :_1 A,  0\Gamma' \vdash }
		  \end{gather*}

		  Since $\Sh(\Gamma) \vdash A : *$, by (2) we have $\pi : \interp{A} \to \und{\interp{\Gamma}}$. Note that there is a canonical morphism $f : p\und{\interp{\Gamma}}\depotimes \interp{A}\depotimes  p\und{\interp{\Gamma'}} \to \interp{A}$ over $\und{\interp{\Gamma}}$. We define $\interp{x}$ as the following
		  $u$ over $\und{\interp{\Gamma, x : A, \Gamma'}}$.
		  \[
		  \begin{tikzcd}
			p\und{\interp{\Gamma}}\depotimes \interp{A}\depotimes  p\und{\interp{\Gamma'}} \arrow[d, dashed, "u"] \arrow[dr, "f"] \arrow[dd, "\eta", bend right=60, swap]&  \\
			\und{\interp{\Gamma'}}\times_{\und{\interp{\Gamma}}} \interp{A} \arrow[r] \arrow[d, "\pi'"] \arrow[dr, phantom, "\LRcorner", very near start]& \interp{A} \arrow[d, "\pi"]  \\
			\und{\interp{\Gamma, x : A, \Gamma'}} = \und{\interp{\Gamma'}} \arrow[r] &
			\und{\interp{\Gamma}}
		  \end{tikzcd}
		  \]

		\item Case:
		  \begin{gather*}
			\infer{\Gamma_1 + \Gamma_2 \vdash M N : B[\Sh(N)]}
			{\Gamma_1 \vdash M : (x : A) \multimap B[x] &
			  \Gamma_2 \vdash N : A}
		  \end{gather*}
		  The assumptions give arrows
		  $\interp{M} : \interp{\Gamma_1} \to \Pi_{\und{\interp{\Gamma}}, \interp{A}} \interp{B}$ and
		  $\interp{N} : \interp{\Gamma_2} \to \interp{A}$
		  of $\EE/\und{\interp{\Gamma}}$ for $\und{\interp{\Gamma}} = \und{\interp{\Gamma_1}} = \und{\interp{\Gamma_2}}$,
		  as well as $\pi: \interp{B} \to \und{\interp{\Sh(\Gamma), x :_k \Sh(A)}} = \und{\interp{A}}$.
		  Note that $\und{\interp{\Sh(N)}} = \und{\interp{N}}$, since Theorem~\ref{app:shape:interp} (3) implies $\interp{\Sh(N)} = p \und{\interp{N}}$.
		  Hence Theorem~\ref{sub:value} gives the pullback square in the following.
		  We therefore define $\interp{M N}$ as the unique arrow $u$ that makes the diagram commute.
		  \begin{gather*}
			\begin{tikzpicture}[x=20pt,y=20pt]
			  \coordinate (O) at (0,0);
			  \coordinate (r) at (6.4,0);
			  \coordinate (d) at (0,-2);
			  \node (A1) [inner sep=0.25em] at (O) {$\interp{\Gamma_1 + \Gamma_2}$};
			  \node (A2) [inner sep=0.25em] at ($ (A1) + (r) $) {$(\Pi_{\und{\interp{\Gamma}}, \interp{A}} \interp{B}) \otimes_{\und{\interp{\Gamma}}} \interp{A}$};
			  \node (B1) [inner sep=0.25em] at ($ (A1) + 0.6*(r) + 0.6*(d) $) {$\interp{B[\Sh(N)]}$};
			  \node (B2) [inner sep=0.25em] at ($ (B1) + (r) $) {$\interp{B}$};
			  \node (C1) [inner sep=0.25em] at ($ (B1) + (d) $) {$\und{\interp{\Gamma}}$};
			  \node (C2) [inner sep=0.25em] at ($ (C1) + (r) $) {$\und{\interp{A}}$};
			  \draw [->] (B1) -- (C1) node [scale=0.7,pos=0.5,inner sep=2pt,right] {$\pi$};
			  \draw [->] (B2) -- (C2);
			  \draw [->] (C1) -- (C2) node [scale=0.7,pos=0.5,inner sep=2pt,below] {$\und{\interp{N}}$};
			  \draw [->] (B1) -- (B2);
			  \draw [->] (A1) -- (A2) node [scale=0.7,pos=0.5,inner sep=2pt,above] {$\interp{M} \otimes_{\und{\interp{\Gamma}}} \interp{N}$};
			  \draw [->] (A2) -- (B2) node [scale=0.7,pos=0.75,anchor=240,inner sep=4pt] {$\epsilon$};
			  \path (A1) edge [->,bend right=15] coordinate [pos=0.5] (A1-C1-label) (C1);
			  \node [scale=0.7,anchor=north east,inner sep=1pt] at (A1-C1-label) {$\eta$};
			  \draw [->,dotted] (A1) -- (B1) node [scale=0.7,pos=0.5,anchor=60,inner sep=2pt] {$u$};
			  \coordinate (B1-pb) at ($ (B1) + (1.8,-0.9) $);
			  \draw ($ (B1-pb) + (-0.35,0) $) -- (B1-pb) -- ($ (B1-pb) + (0,0.35) $);
			\end{tikzpicture}
		  \end{gather*}
		  It is crucial to observe that $\interp{\Gamma_1+\Gamma_2} = \interp{\Gamma_1}\depotimes \interp{\Gamma_2}$.
		  This then enables us to use $\epsilon : (\Pi_{\und{\interp{\Gamma}}, \interp{A}} \interp{B}) \otimes_{\und{\interp{\Gamma}}} \interp{A}\to \interp{B}$,
		  the counit of the adjunction of Theorem~\ref{adj:dep}.

		\item Case:
		  \begin{gather*}
			\infer{\Gamma \vdash \lambda x . M : (x : A) \multimap B[x]}
            {\Gamma, x :_k A \vdash M : B[x] & k = 0 \Rightarrow A = P}
		  \end{gather*}

		  By the induction hypothesis, we have a morphism $\interp{M} : \interp{\Gamma}\otimes_{\und{\interp{\Gamma}}} \interp{A}^k \to \interp{B}$
		  over $\und{\interp{\Gamma, x : A}}$. By the adjunction in Theorem~\ref{adj:dep},
		  we define \[\interp{\lambda x. M} = \widetilde{\interp{M}} : \interp{\Gamma} \to \prod_{\und{\interp{\Gamma}}, \interp{A}^k} \interp{B}\] as a morphism over $\und{\interp{\Gamma}}$.

		\item Case:
		  \begin{gather*}
			\infer{\Phi \vdash \lambda' x . R : (x : P_1) \to P_2[x]}
            {\Phi, x : P_1 \vdash R : P_2[x]} 
		  \end{gather*}
		  By induction, we have $\interp{R} : \interp{\Phi} \otimes_{\und{\interp{\Phi}}} \interp{P_1} \to \interp{P_2}$.
		  By Theorem~\ref{app:shape:interp} (3), we have $\interp{\Sh R} = \interp{R} = p\und{\interp{R}}$, thus
		  $\interp{R}$ is a morphism in $p \mathbf{B}/\und{\interp{P_1}}$.
		  We therefore define $\interp{\lambda' x . R}$ to be the unique
		  morphism $p(\und{\interp{\Phi}} \to \und{\prod}_{\und{\interp{\Phi}}, \und{\interp{P_1}}} \und{\interp{P_2}})$ over $\und{\interp{\Phi}}$.

		\item Case:
		  \begin{gather*}
			\infer{\Phi \vdash R_1 @ R_2 : P_2[R_2]}
			{\Phi \vdash R_1 : (x : P_1) \to P_2[x] &
			  \Phi \vdash R_2 : P_1}
		  \end{gather*}

		  By induction and Theorem~\ref{app:shape:interp} (3), we have
		  \[\interp{R_1} = p\und{\interp{R_1}} : p\und{\interp{\Phi}} \to p(\und{\prod}_{\und{\interp{\Phi}}, \und{\interp{P_1}}} \und{\interp{P_2}})\]
		  and  \[ \interp{R_2} = p\und{\interp{R_2}} : p\und{\interp{\Phi}} \to p\und{\interp{P_1}}.\]
		  Both are morphisms in $p\mathbf{B}/\und{\interp{\Phi}}$.
		  We define $\interp{R_1 @ R_2}$ as the following arrow $u$ (where
		  $\epsilon : (\und{\prod}_{\und{\interp{\Phi}}, \und{\interp{P_1}}} \und{\interp{P_2}}) \times_{\und{\interp{\Phi}}} \und{\interp{P_1}} \to \und{\interp{P_2}}$ is the counit in $\BB$).
		  Note that $u$ is also a morphism in $p\mathbf{B}/\und{\interp{\Phi}}$.

		  \[
		  \begin{tikzcd}
			p\und{\interp{\Phi}}\arrow[d, dashed, "u"] \arrow[dr, "p(\epsilon \circ (\und{\interp{R_1}} \times_{\und{\interp{\Phi}}} \und{\interp{R_2}}))"] \arrow[dd, "\eta", bend right=80, swap]&  \\
			p\und{\interp{P_2[R_2]}} \arrow[r] \arrow[d] \arrow[dr, phantom, "\LRcorner", very near start]& p\und{\interp{P_2}} \arrow[d, "\pi"]  \\
			\und{\interp{\Phi}} \arrow[r, "\und{\interp{R_2}}"] &  \und{\interp{\Phi, x : P_1}}
		  \end{tikzcd}
		  \]

		\item Case:
		  \begin{gather*}
			\infer{\Phi \vdash \liftt M : {!}A}{\Phi \vdash M : A}
		  \end{gather*}

		  By Theorem~\ref{app:shape:interp} (1),
		  we have $\interp{\Phi} = p\und{\interp{\Phi}}$.
		  By the induction hypothesis, we have a morphism $\interp{M} : p\und{\interp{\Phi}} \to \interp{A}$ over $\und{\interp{\Phi}}$. By the adjunction $p \dashv \flat$,
		  we define $\interp{\liftt M} = p\widetilde{\interp{M}} : p\und{\interp{\Phi}} \to p\flat \interp{A}$ as a morphism over $\und{\interp{\Phi}}$.

		\item Case:
		  \begin{gather*}
			\infer{\Gamma \vdash \force M : A}{\Gamma \vdash M : {!}A}
		  \end{gather*}

		  By the induction hypothesis, we have a morphism $\interp{M} : \interp{\Gamma} \to p\flat\interp{A}$ over
		  $\und{\interp{\Gamma}}$. We define $\interp{\force M} = \force \circ \interp{M} : \interp{\Gamma} \to \interp{A}$, where $\force : {!} \to \id_{\EE}$ is the counit.

		\item Case:
		  \begin{gather*}
			\infer{\Phi \vdash \forceprime R : \Sh(A)}{\Phi \vdash R : {!}A}
		  \end{gather*}

		  By induction and Theorem~\ref{app:shape:interp} (3) we have a morphism $\interp{R} : \interp{\Phi} \to p\flat\interp{A}$
		  in $p\mathbf{B}/\und{\interp{\Phi}}$.
		  We define $\interp{\forceprime R} = p\sharp\force\circ \interp{R} : \interp{\Phi} \to p\sharp\interp{A}$ as
		  a morphism in $p\mathbf{B}/\und{\interp{\Phi}}$.
		\item Case:
		  \begin{gather*}
			\infer{\Gamma_1 + \Gamma_2 \vdash (M, N) :  (x : A) \otimes B[x]}{\Gamma_1 \vdash M : A & \Gamma_2 \vdash N : B[\Sh(M)]}
		  \end{gather*}

		  By the induction hypothesis, $\interp{M} : \interp{\Gamma_1} \to \interp{A}$
		  and $\interp{N} : \interp{\Gamma_2} \to \interp{B[\Sh(M)]}$.
		  Since $\Sh(\Gamma_1) \vdash \Sh(M) : \Sh(A)$,
		  by Theorem~\ref{app:shape:interp} (3),
		  we have $p\und{\interp{M}} : p\und{\interp{\Gamma_1}} \to p\und{{\interp{A}}}$,
		  thus $\und{p \interp{M}} = \und{\interp{M}} : \und{\interp{\Gamma}} \to \und{\interp{A}}$.
		  By Theorem~\ref{type:pullback}, we have a map $i : \interp{B[\Sh(M)]} \to \interp{B}$. We define
		  \[\interp{(M, N)} = \interp{M}\depotimes (i \circ \interp{N}) : \interp{\Gamma_1}\depotimes \interp{\Gamma_2} \to
		  \interp{A}\depotimes \interp{B}.\]

		\item Case:
		  \begin{gather*}
			\infer
            {\Gamma_1+ \Gamma_2 \vdash \lett (x, y) = N \tin M : C}
            {
              \begin{array}{l}
                \Gamma_1 \vdash N : (x : A) \otimes B[x] \\
                \Gamma_2, x :_{k_1} A, y :_{k_2} B[x] \vdash M : C\\
                k_1 = 0 \Rightarrow A = P_1, k_2 = 0 \Rightarrow B[x] = P_2[x]
			  \end{array}}
		  \end{gather*}

		  Without loss of generality, we only consider $k_1 = k_2 = 1$.
		  By the induction hypothesis, we have $\interp{N} : \interp{\Gamma_1} \to \interp{A}\otimes_{\und{\interp{A}}} \interp{B}$ over $\und{\interp{\Gamma}}$ and
		  $\interp{M} : \interp{\Gamma_2, x :_{1} A, y :_{1} B[ x]} \to \interp{C}$ over
		  $\und{\interp{B}}$.
		  So we have $\pi : \interp{C} \to \und{\interp{B}}$.
		  Moreover, we have $\und{\interp{N}} : \und{\interp{\Gamma}} \to \und{\interp{B}}$. So $\und{\interp{N}}$ and $\pi$ form a pullback square with pullback $\interp{C}$ over $\und{\interp{\Gamma}}$.

		  We have \[\interp{\Gamma_2} \depotimes \interp{N} : \interp{\Gamma_1 + \Gamma_2} \to \interp{\Gamma_2} \depotimes (\interp{A}\depotimes \interp{B}) = \interp{\Gamma_2, x :_1 A, y :_1 B[x]}.\]
		  Thus $\interp{M} \circ (\interp{\Gamma_2} \depotimes \interp{N}) : \interp{\Gamma_1 + \Gamma_2} \to \interp{C}$.

		  We define $\interp{\lett ( x, y ) = M \tin N}$ to be the unique arrow over $\und{\interp{\Gamma}}$ induced by the pullback.
		  \[
		  \begin{tikzcd}
			\interp{\Gamma_1+\Gamma_2} \arrow[d, dashed, "\interp{let}"] \arrow[dr, "\interp{M} \circ (\interp{\Gamma_2} \depotimes \interp{N})"] \arrow[dd, "\eta", bend right=70, swap]&  \\
			\und{\interp{\Gamma}}\times_{\und{\interp{B}}} \interp{C} \arrow[r] \arrow[d] \arrow[dr, phantom, "\LRcorner", very near start]& \interp{C} \arrow[d, "\pi"]  \\
			\und{\interp{\Gamma}} \arrow[r, "\und{\interp{N}}"] &  \und{\interp{B}}
		  \end{tikzcd}
		  \qedhere
		  \]

	  \end{itemize}
  \end{enumerate}
\end{proof}

\begin{theorem}
  \label{app:shape:interp}
  \begin{enumerate}
	\item If $\Gamma \vdash$, then we have $\interp{\Sh(\Gamma)} = p\und{\interp{\Gamma}}$, an object in $p\BB$.
	\item If $\Phi \vdash A : *$, then we have $\interp{\Phi \vdash \Sh(A)} = \eta \circ p\und{\pi} : p\und{\interp{A}} \to \und{\interp{\Phi}}$, an object in the slice category $p\BB/\und{\interp{\Phi}}$, where $\eta : p\und{\interp{\Phi}} \to \und{\interp{\Phi}}$
	  and $\pi : \interp{A} \to \und{\interp{\Phi}}$.
	\item If $\Gamma \vdash M : A$, then we have $\interp{\Sh(\Gamma) \vdash \Sh(M) : \Sh(A)}
	  = p\und{\interp{M}} : p\und{\interp{\Gamma}} \to p\und{\interp{A}}$, a morphism in the slice category $p\BB/\und{\interp{\Gamma}}$.
  \end{enumerate}
\end{theorem}
\begin{proof}
  \begin{enumerate}
	\item  Since $\interp{\cdot \vdash} = I = p 1 = p\sharp I$, we only need to consider the following case.
      
      \begin{gather*}
        \infer{\Gamma, x :_k A \vdash }
		{\Gamma \vdash & \Sh(\Gamma) \vdash A : *}
      \end{gather*}
      By the induction hypothesis, we have $\interp{\Sh(\Gamma)} = p\und{\interp{\Gamma}}$.
      By (2), we have a map $p\und{\interp{A}} \to \und{\interp{\Gamma}}$.
      We have
      $\interp{\Sh(\Gamma), x :_k \Sh(A)} = p\und{\interp{\Gamma}}\otimes_{\und{\interp{\Gamma}}} p\und{\interp{A}} = p\sharp (\interp{\Gamma}\depotimes \interp{A}) = p\und{\interp{\Gamma, x :_k A}}$
      as a fibered tensor product over $\und{\interp{\Gamma}}$. 

	\item By induction on derivation of $\Phi \vdash A : *$.

	  \begin{itemize}
		\item Case:
		  \begin{gather*}
			\infer{\Phi \vdash {!}A : *}{\Phi \vdash A : *}
		  \end{gather*}

		  We have
		  $\pi = \interp{\Phi \vdash {!}A} =  q\flat \pi' \circ \eta_{p\flat \interp{A}}  : p \flat \interp{A} \to \und{\interp{\Phi}}$, where $\pi' : \interp{A} \to q\und{\interp{\Phi}}$.
		  We need to show that $\pi = \eta \circ p \sharp \pi$.
                  By Fact~\ref{term-fibered} (5), we have
                  $\pi = \eta \circ p \tilde{\pi} = \eta \circ p \sharp \pi$. 
		\item Case:
		  \begin{gather*}
			\infer{\Phi \vdash (x : A_1) \multimap A_2[x] : *}
			{\Phi, x :_k A_1 \vdash A_2[x] : *}
		  \end{gather*}

		  Since we also have $\Phi \vdash (x : A_1) \multimap A_2[x] : *$,
		  by Theorem~\ref{app:interp} (2), we have
                  \[\pi = \interp{x: A_{1} \multimap A_{2}[x]} : \Pi_{\und{\interp{\Phi}},\interp{A_1}} \interp{A_2} \to \und{\interp{\Phi}}. \]
                  Also by Theorem~\ref{app:interp} (2), for
		  the kinding of $\Phi \vdash (x : \Sh(A_1)) \to \Sh(A_2[x])$, we have
		  \[\eta_{p\und{\interp{\Phi}}}\circ p p_1 : p(\Pi_{\und{\interp{\Phi}}, \und{\interp{A_1}}} \und{\interp{A_2}}) \to \und{\interp{\Phi}},\]
                  where $p_1 : \und{\Pi}_{\und{\interp{\Phi}}, \und{\interp{A_1}}} \und{\interp{A_2}} \to \und{\interp{\Phi}}$.
                  Since $p_{1} = \sharp \pi$, we have
                  \[\interp{\Phi \vdash (x : \Sh(A_1)) \to
                    \Sh(A_2[x])} = \eta \circ p \sharp \pi.
                  \]
	  \end{itemize}
	\item By induction on the derivation of $\Gamma \vdash M : A$. Here we show a few cases.

	  \begin{itemize}
		\item Case:
		  \begin{gather*}
			\infer{0\Gamma, x :_1 A, 0\Gamma' \vdash x : A}
			{0\Gamma, x :_1 A,  0\Gamma' \vdash }
		  \end{gather*}

		  Since $\Sh(\Gamma) \vdash \Sh(A) : *$
		  and $\Sh(\Gamma) \vdash A : *$, by (2) we have $\pi : \interp{A} \to \und{\interp{\Gamma}}$ and $\eta \circ p\sharp \pi : p\und{\interp{A}} \to \und{\interp{\Gamma}}$.
		  Note that there is a canonical morphism $f : \interp{0\Gamma, x :_1 A,  0\Gamma'} \to \interp{A}$ over $\und{\interp{\Gamma}}$, also $p\sharp f : \interp{0\Gamma, x :_1 \Sh(A),  0\Gamma'} \to p\und{\interp{A}}$. We have the following diagram (note that $p\sharp (\und{\interp{\Gamma'}}\times_{\und{\interp{\Gamma}}} \interp{A}) = \und{\interp{\Gamma'}}\times_{\und{\interp{\Gamma}}} p\und{\interp{A}}$).
		  \[
		  \begin{tikzcd}
			\interp{0\Gamma, x :_1 \Sh(A), 0\Gamma'} \arrow[d, dashed, "u"] \arrow[dr, "p\sharp f"] & \\
			p\sharp (\und{\interp{\Gamma'}}\times_{\und{\interp{\Gamma}}} \interp{A}) = \und{\interp{\Gamma'}}\times_{\und{\interp{\Gamma}}} p\und{\interp{A}} \arrow[r] \arrow[d] \arrow[dr, phantom, "\LRcorner\quad", very near start]& p\und{\interp{A}} \arrow[d, "\eta \circ p\sharp \pi"]  \\
			\und{\interp{\Gamma, x : A, \Gamma'}} \arrow[r] &
			\und{\interp{\Gamma}}
		  \end{tikzcd}
		  \]

		  Thus $u = p\und{\interp{x}}'$, where
		  \[\interp{x}' = \interp{0\Gamma, x :_1 A, 0\Gamma' \vdash x : A} : \interp{0\Gamma, x :_1 A, 0\Gamma'} \to \und{\interp{\Gamma'}}\times_{\und{\interp{\Gamma}}} \interp{A}.\]

		\item Case:
		  \begin{gather*}
			\infer{\Gamma_{1}+\Gamma_{2} \vdash M N :  B[\Sh(N)]}
			{\Gamma_{1} \vdash M : (x : A) \multimap B[x] &
			  \Gamma_{2} \vdash N : A}
		  \end{gather*}
		  By the induction hypothesis, we have
		  \[\interp{\Sh(N)} = p\und{\interp{N}} : p\und{\interp{\Gamma}} \to p\und{\interp{A}}\]
		  and \[\interp{\Sh(M)} = p\und{\interp{M}} : p\und{\interp{\Gamma}} \to p\sharp (\Pi_{\und{\interp{\Gamma}}, \interp{A}}\interp{B}).\]
		  Thus $p\sharp (\epsilon \circ (\interp{M}\otimes_{\und{\interp{\Gamma}}} \interp{N})) =
		  p\sharp \epsilon \circ (p\und{\interp{M}}\otimes_{\und{\interp{\Gamma}}} p\und{\interp{N}})$.
		  Since $p\sharp\interp{B[\Sh(N)]} = \interp{\Sh(B[\Sh(N)])}$, we have $p\und{\interp{MN}} = \interp{\Sh(M) @ \Sh(N)}$.

		\item Case:
		  \begin{gather*}
			\infer{\Gamma \vdash \lambda x . M : (x : A) \to B[x]}{\Gamma, x :_k A \vdash M : B[x] & k = 0 \Rightarrow A = P}
		  \end{gather*}

		  By the induction hypothesis, we have $\interp{\Sh(M)} = p\und{\interp{M}} : p\und{\interp{\Gamma}}\otimes_{\und{\interp{\Gamma}}} p\und{\interp{A}} \to p\und{\interp{B}}$.
		  Thus
		  \[\interp{\lambda' x . \Sh(M)} = p\sharp\interp{\lambda x.M} : p\und{\interp{\Gamma}} \to p(\und{\Pi}_{\und{\interp{\Gamma}}, \und{\interp{A}}} \und{\interp{B}})\]
		  by uniqueness of the abstraction.

		\item Case.
		  \begin{gather*}
			\infer{\Gamma \vdash \force M : A}{\Gamma \vdash M : {!}A}
		  \end{gather*}

		  By the induction hypothesis, we have $\interp{\Sh(M)} = p\und{\interp{M}} : p\und{\interp{\Gamma}} \to p\flat\interp{A}$.
		  We have
		  \begin{gather*}
		  p\sharp \interp{\force M} = p\sharp \force \circ p\sharp \interp{M} = p\sharp \force \circ \interp{\Sh(M)} : p\und{\interp{\Gamma}} \to p\und{\interp{A}}.
		  \qedhere
		  \end{gather*}
	  \end{itemize}

  \end{enumerate}
\end{proof}

Next we will prove two semantics substitution theorems, one for types and the other for terms.
\begin{theorem}\label{thm-dagger}
  Suppose $\Gamma, x :_k A, \Gamma' \vdash$ and $\Gamma \vdash N : A$ and $\Gamma, [\Sh(N)/x]\Gamma' \vdash$. We have the following pullback.
  \[
    \begin{tikzcd}
        \interp{[\Sh(N)/x]\Gamma'} \arrow[r, "\interp{N}^\dagger"]
        \arrow[dr, phantom, "\LRcorner", very near start]
        \arrow[d]
        & \interp{\Gamma'} \arrow[d] \\
        \und{\interp{\Gamma}} \arrow[r, "\und{\interp{N}}"] &
        \und{\interp{A}}
    \end{tikzcd}
  \]
\end{theorem}
\begin{proof}
  Consider the following case.

  \begin{gather*}
    \infer{\Gamma, x :_k A, \Gamma', y :_k B \vdash}{\Gamma, x :_k A, \Gamma' \vdash & \Gamma, x :_k A, \Gamma' \vdash B : *}
  \end{gather*}

  \noindent By the induction hypothesis, we have the following pullback on the left. By Theorem~\ref{type:pullback}, we have the pullback on the right.
  \[
    \begin{tabular}{lll}
    \begin{tikzcd}
        \interp{[\Sh(N)/x]\Gamma'} \arrow[r]
        \arrow[dr, phantom, "\LRcorner", very near start]
        \arrow[d]
        & \interp{\Gamma'} \arrow[d] \\
        \und{\interp{\Gamma}} \arrow[r, "\und{\interp{N}}"] &
        \und{\interp{A}}
    \end{tikzcd}
    &
    \begin{tikzcd}
        \interp{[\Sh(N)/x]B} \arrow[r]
        \arrow[dr, phantom, "\LRcorner", very near start]
        \arrow[d, "\pi'"]
        & \interp{B} \arrow[d, "\pi"] \\
        \und{\interp{[\Sh(N)/x]\Gamma'}} \arrow[r, "\und{\interp{N}}^\dagger"] &
        \und{\interp{\Gamma'}} 
    \end{tikzcd}
    \end{tabular}
  \]

\noindent By Lemma~\ref{stack:pullback}, we have the following pullback.
  \[
      \begin{tikzcd}
        \interp{[\Sh(N)/x]\Gamma'} \depotimes \interp{[\Sh(N)/x]B} \arrow[r]
        \arrow[dr, phantom, "\LRcorner", very near start]
        \arrow[d, "\pi'"]
        & \interp{\Gamma'} \depotimes\interp{B} \arrow[d, "\pi"] \\
        \und{\interp{\Gamma}} \arrow[r, "\und{\interp{N}}"] &
        \und{\interp{A}} 
      \end{tikzcd}
    \qedhere
  \]
\end{proof}

\begin{remark}
  Theorem~\ref{thm-dagger} also serves as the definition of the
  morphism $\interp{N}^{\dagger}$, which we will use later.  Also, in
  the special case where $N=R$ is a parameter term, $A=P$ is a
  parameter type, and $\Gamma=\Phi$ and $\Gamma'=\Phi'$ are parameter
  contexts, the pullback in Theorem~\ref{thm-dagger} takes the simpler
  form
  \[
    \begin{tikzcd}
        \interp{[R/x]\Phi'} \arrow[r, "\interp{R}^\dagger"]
        \arrow[dr, phantom, "\LRcorner", very near start]
        \arrow[d]
        & \interp{\Phi'} \arrow[d] \\
        \und{\interp{\Phi}} \arrow[r, "\und{\interp{R}}"] &
        \und{\interp{P}}
    \end{tikzcd}
  \]
\end{remark}

\begin{lemma}
  \label{Peter's-lemma}
Suppose we have a pullback
\[
      \begin{tikzcd}
	A'
	\arrow[r, "N^*"]
	\arrow[dr, phantom, "\LRcorner", very near start]
	\arrow[d]
	&
	A
	\arrow[d]
	\\
	qY
	\arrow[r, "N"]
	&
	qX
      \end{tikzcd}
\]
in $\EE$.
Then, for any object $C$ in $\EE / qY$, we have $C \otimes_{qY} A' \cong C \otimes_{qX} A$.
\end{lemma}

\begin{proof}
  First observe in $\BB$ that, by \eqref{term-fibered.3} of Fact~\ref{term-fibered} and the supposition, the right square below is a pullback, and hence that the outer square below is a pullback by the pullback lemma;
  i.e., that $\und{C} \times_Y \und{A'} \cong \und{C} \times_X \und{A}$.
\[
      \begin{tikzcd}
	\und{C} \times_Y \und{A'}
	\arrow[d, "\pi_0", swap]
	\arrow[r, "\pi_1"]
	\arrow[dr, phantom, "\LRcorner", very near start]
	&
	\und{A'}
	\arrow[r, "\und{N^\ast}"]
	\arrow[dr, phantom, "\LRcorner", very near start]
	\arrow[d]
	&
	\und{A}
	\arrow[d]
	\\
	\und{C}
	\arrow[r]
	&
	Y
	\arrow[r, "N"]
	&
	X
      \end{tikzcd}
\]
Now, $C \otimes_Y A'$, on the one hand, is obtained from $C \otimes A'$ by the left pullback below with the monic $i = \langle \pi_0, \pi_1 \rangle$ for $\pi_0$ and $\pi_1$ in the diagram above, while the right square is a pullback by \eqref{state-parameter.1} of Definition~\ref{state-parameter}.
\[
      \begin{tikzcd}
	C \otimes_{qY} A'
	\arrow[d, "\eta", swap]
	\arrow[r]
	\arrow[dr, phantom, "\LRcorner", very near start]
	&
	C \otimes A'
	\arrow[r, "1_C \otimes N^\ast"]
	\arrow[dr, phantom, "\LRcorner", very near start]
	\arrow[d, "\eta", swap]
	&
	C \otimes A
	\arrow[d, "\eta"]
	\\
	\und{C} \times_Y \und{A'}
	\arrow[r, "i"]
	&
	\und{C} \times \und{A'}
	\arrow[r, "1_{\und{C}} \times \und{N^\ast}"]
	&
	\und{C} \times \und{A}
      \end{tikzcd}
\]
On the other hand, $C \otimes_X A$ is obtained from $C \otimes A$ by the pullback below.
But the bottom arrows of the outer pullback above (i.e., $(1_{\und{C}} \times \und{N^\ast}) \circ \langle \pi_0, \pi_1 \rangle$) and of the pullback below (i.e., $j = \langle \pi_0, \und{N^\ast} \circ \pi_1 \rangle$) are the same arrow by the first observation we made above.
\[
      \begin{tikzcd}
	C \otimes_{qX} A
	\arrow[r]
	\arrow[dr, phantom, "\LRcorner", very near start]
	\arrow[d, "\eta", swap]
	&
	C \otimes A
	\arrow[d, "\eta"]
	\\
	\und{C} \times_X \und{A}
	\arrow[r, "j"]
	&
	\und{C} \times \und{A}
      \end{tikzcd}
  \qedhere
\]
\end{proof}

\begin{theorem}
  \label{type:pullback}
  If $\Phi, x : P, \Phi' \vdash B : *$ and $\Phi \vdash R : P$, then $\Phi, [R/x]\Phi' \vdash [R/x]B : *$.
  Semantically, suppose we have $\pi : \interp{B} \to \und{\interp{\Phi, x : P, \Phi'}}$,
 a morphism $\interp{R} : \interp{\Phi} \to \interp{P}$ over $\und{\interp{\Phi}}$,
 and $\pi' : \interp{[R/x]B} \to \interp{[R/x]\Phi'}$, and a morphism $\interp{R}^\dagger : \interp{[R/x]\Phi'} \to \interp{\Phi'}$. Then $\pi'$ is the pullback of $\pi$ along $\und{\interp{R}}^\dagger$, i.e., we have the following pullback square.
  \[
    \begin{tikzcd}
        \interp{[R/x]B} \arrow[r]
        \arrow[dr, phantom, "\LRcorner", very near start]
        \arrow[d, "\pi'", swap]
        & \interp{B} \arrow[d, "\pi"] \\
        \und{\interp{[R/x]\Phi'}} \arrow[r, "\und{\interp{R}}^\dagger"] &
        \und{\interp{\Phi'}}
    \end{tikzcd}
  \]
\end{theorem}

\begin{proof}
  By induction on the derivation of $\Phi, x : P, \Phi' \vdash B : *$.

  \begin{itemize}
  \item Case:
    \begin{gather*}
      \infer{\Phi, x : P, \Phi' \vdash {!}A : *}{\Phi, x : P, \Phi' \vdash A : *}
    \end{gather*}

    Let us write $\sigma$ for the substitution $[R/x]$.
    By the induction hypothesis, we have the following left pullback, which
    gives rise to the right pullback diagram.
    \[
      \begin{tabular}{llll}
          \begin{tikzcd}
            \interp{\sigma A} \arrow[r, "a"]
            \arrow[dr, phantom, "\LRcorner", very near start]
            \arrow[d, "\pi_{\sigma A}", swap]
            & \interp{A} \arrow[d, "\pi_A"] \\
            \und{\interp{\sigma \Phi}} \arrow[r, "\und{\interp{R}}^\dagger"]
            & \und{\interp{\Phi}}
          \end{tikzcd}
        & &
          \begin{tikzcd}
            p\flat\interp{\sigma A} \arrow[r, "p\flat a"]
            \arrow[dr, phantom, "\LRcorner", very near start]
            \arrow[d, "\eta_{p\flat\interp{\sigma A}}", swap]
            & p\flat\interp{A} \arrow[d, "\eta_{p\flat\interp{A}}", swap] \\
            q\flat\interp{\sigma A} \arrow[r, "q\flat a"]
            \arrow[dr, phantom, "\LRcorner", very near start]
            \arrow[d, "q\flat \pi_{\sigma A}", swap]
            & q\flat\interp{A} \arrow[d, "q\flat \pi_A"] \\
            \und{\interp{\sigma \Phi}} \arrow[r, "\und{\interp{R}}^\dagger"]
            & \und{\interp{\Phi}}
          \end{tikzcd}
      \end{tabular}
    \]

    Since $\interp{\Phi, \sigma\Phi' \vdash {!} \sigma A : *} = q \flat \pi_{\interp{\sigma A}}\circ \eta_{q\flat \interp{\sigma A}}$,
    the right diagram above is the required pullback.

  \item Case:
    \begin{gather*}
    \infer{\Phi, y : P, \Phi' \vdash (x : A) \otimes B[x] : *}{\Phi, y : P, \Phi', x :_k \Sh(A) \vdash B[x] : *}
    \end{gather*}

    By the induction hypothesis, we have the following two pullbacks.
  \[
    \begin{tabular}{lll}
    \begin{tikzcd}
        \interp{[R/y]B} \arrow[r]
        \arrow[dr, phantom, "\LRcorner", very near start]
        \arrow[d]
        & \interp{B} \arrow[d] \\
        \und{\interp{[R/y]A}} \arrow[r, "\und{\interp{R}^\dagger}"] &
        \und{\interp{A}}
    \end{tikzcd}
    &
    \begin{tikzcd}
        \interp{[R/y]A} \arrow[r]
        \arrow[dr, phantom, "\LRcorner", very near start]
        \arrow[d]
        & \interp{A} \arrow[d] \\
        \und{\interp{[R/y]\Phi'}} \arrow[r, "\und{\interp{R}}^\dagger"] &
        \und{\interp{\Phi'}}
    \end{tikzcd}
    \end{tabular}
      \]

      Thus we have the following pullback.
  \[
    \begin{tikzcd}
      \interp{[R/y]A}\depotimes
      \interp{[R/y]B}
      \arrow[r]
        \arrow[dr, phantom, "\LRcorner\quad", very near start]
        \arrow[d]
        & \interp{A}\depotimes\interp{B} \arrow[d] \\
        \und{\interp{[R/y]\Phi'}} \arrow[r, "\und{\interp{R}}^\dagger"] &
        \und{\interp{\Phi'}}
    \end{tikzcd}
  \]

  \item Case:
    \begin{gather*}
      \infer{\Phi, y : P, \Phi' \vdash (x : A) \multimap B[x] : *}{\Phi, y : P, \Phi', x :_k \Sh(A) \vdash B[x] : *}
    \end{gather*}

    Let us write $\sigma$ for the substitution $[R/y]$.
        By the induction hypothesis, we have the first two pullbacks below, which make the third a pullback, too.
  \[
    \begin{tabular}{lll}
    \begin{tikzcd}
        \interp{\sigma B} \arrow[r]
        \arrow[dr, phantom, "\LRcorner", very near start]
        \arrow[d]
        & \interp{B} \arrow[d] \\
        \und{\interp{\sigma A}} \arrow[r, "\und{\interp{R}^\dagger}"] &
        \und{\interp{A}}
    \end{tikzcd}
    &
    \begin{tikzcd}
        \interp{\sigma A} \arrow[r]
        \arrow[dr, phantom, "\LRcorner", very near start]
        \arrow[d]
        & \interp{A} \arrow[d] \\
        \und{\interp{\sigma \Phi'}} \arrow[r, "\und{\interp{R}}^\dagger"] &
        \und{\interp{\Phi'}}
    \end{tikzcd}
    &
    \begin{tikzcd}
        \interp{\sigma B} \arrow[r]
        \arrow[dr, phantom, "\LRcorner", very near start]
        \arrow[d]
        & \interp{B} \arrow[d] \\
        \und{\interp{\sigma \Phi'}} \arrow[r, "\und{\interp{R}}^\dagger"] &
        \und{\interp{\Phi'}}
    \end{tikzcd}
    \end{tabular}
  \]

  We first claim that the pullback $D$ in
\[
    \begin{tikzcd}
        D
        \arrow[r]
        \arrow[dr, phantom, "\LRcorner", very near start]
        \arrow[d]
        &\interp{A}\multimap_{\und{\interp{\Phi'}}} \interp{B}
        \arrow[d] \\
        \und{\interp{\sigma\Phi'}} \arrow[r, "\und{\interp{R}}^\dagger"] &
        \und{\interp{\Phi'}}
    \end{tikzcd}
\]
equals $\interp{\sigma A}\multimap_{\und{\interp{\sigma \Phi'}}} \interp{\sigma B}$.

The claim is proved as follows.
An arrow $h$ that makes
\[
\begin{tikzpicture}[x=25pt,y=25pt]
\coordinate (O) at (0,0);
\foreach \x in {0,1,2} {
	\node [outer sep=0pt,inner sep=2pt] (v\x) at ({cos((-4 * \x + 5)*pi/6 r)},{sin((-4 * \x + 5)*pi/6 r)}) {};
};
\node (A) [inner sep=0.25em] at (v0) {$C$};
\node (B) [inner sep=0.25em] at (v1) {$D$};
\node (X) [inner sep=0.25em] at (v2) {$\und{\interp{\sigma\Phi'}}$};
\draw [->] (A) -- (B) node [scale=0.7,pos=0.5,inner sep=4pt,above] {$h$};
\draw [->] (A) -- (X) node [scale=0.7,pos=0.5,inner sep=4pt,left] {$c$};
\draw [->] (B) -- (X);
\end{tikzpicture}
\]
commute corresponds one-to-one to an $h'$ that makes
  \[
    \begin{tikzcd}
	C
	\arrow[r, "h'"]
	\arrow[d, "c", swap]
	&
	\interp{A} \multimap_{\und{\interp{\Phi'}}} \interp{B}
	\arrow[d]
	\\
	\und{\interp{\sigma\Phi'}}
	\arrow[r, "\und{\interp{R}}^\dagger"]
	&
	\und{\interp{\Phi'}}
    \end{tikzcd}
  \]
commute, by the definition of $D$ as the pullback above.
Such an $h'$ in turn corresponds one-to-one to a $k'$ that makes
  \[
    \begin{tikzcd}
	C \otimes_{\und{\interp{\Phi'}}} \interp{A}
	\arrow[r, "k'"]
	\arrow[d]
	&
	\interp{B}
	\arrow[d]
	\\
	\und{\interp{\sigma\Phi'}}
	\arrow[r, "\und{\interp{R}}^\dagger"]
	&
	\und{\interp{\Phi'}}
    \end{tikzcd}
  \]
commute, by the adjunction ${-} \otimes_{\und{\interp{\Phi'}}} \interp{A} \dashv \interp{A} \multimap_{\und{\interp{\Phi'}}} {-}$.
But the right and bottom arrows here have a pullback and that is $\interp{\sigma B}$ as seen above---therefore such a $k'$ corresponds to a $k$ that makes
\[
\begin{tikzpicture}[x=30pt,y=30pt]
\coordinate (O) at (0,0);
\foreach \x in {0,1,2} {
	\node [outer sep=0pt,inner sep=2pt] (v\x) at ({cos((-4 * \x + 5)*pi/6 r)},{sin((-4 * \x + 5)*pi/6 r)}) {};
};
\node (A) [inner sep=0.25em] at (v0) {$C \otimes_{\und{\interp{\Phi'}}} \interp{A}$};
\node (B) [inner sep=0.25em] at (v1) {$\interp{\sigma B}$};
\node (X) [inner sep=0.25em] at (v2) {$\und{\interp{\sigma\Phi'}}$};
\draw [->] (A) -- (B) node [scale=0.7,pos=0.5,inner sep=4pt,above] {$k$};
\draw [->] (A) -- (X);
\draw [->] (B) -- (X);
\end{tikzpicture}
\]
commute, where $C \otimes_{\und{\interp{\Phi'}}} \interp{A} \cong C \otimes_{\und{\interp{\sigma \Phi'}}} \interp{\sigma A}$ by Lemma~\ref{Peter's-lemma}.
Thus we have a correspondence
\begin{gather*}
\EE / \und{\interp{\sigma \Phi'}}(C, D)
\cong \EE / \und{\interp{\sigma \Phi'}}(C \otimes_{\und{\interp{\sigma \Phi'}}} \interp{\sigma A}, \interp{\sigma B})
\cong \EE / \und{\interp{\sigma \Phi'}}(C, \interp{\sigma A} \multimap_{\und{\interp{\sigma \Phi'}}} \interp{\sigma B})
\end{gather*}
by the adjunction ${-} \otimes_{\und{\interp{\sigma \Phi'}}} \interp{\sigma A} \dashv \interp{\sigma A} \multimap_{\und{\interp{\sigma \Phi'}}} {-}$, and it is easy to check that the correspondence is natural in $C$.
Therefore $D \cong \interp{\sigma A}\multimap_{\und{\interp{\sigma \Phi'}}} \interp{\sigma B}$ by the Yoneda principle. This finishes the proof of the claim.

From this, it now follows that the following is a pullback (we write $X$ for $q X$ since $q$ is a full embedding).
\[
    \begin{tikzcd}
	\interp{\sigma A} \multimap_{\und{\interp{\sigma \Phi'}}} \interp{\sigma B}
	\arrow[r]
	\arrow[dr, phantom, "\LRcorner", very near start]
	\arrow[d, "\eta", swap]
	&
	\interp{A} \multimap_{\und{\interp{\Phi'}}} \interp{B}
	\arrow[d, "\eta"]
	\\
	\und{\interp{\sigma A}} \Rightarrow_{\und{\interp{\sigma \Phi'}}} \und{\interp{\sigma B}}
	\arrow[r]
	&
	\und{\interp{A}} \Rightarrow_{\und{\interp{\Phi'}}} \und{\interp{B}}
    \end{tikzcd}
\]

Now, consider the following cube, where all the vertical arrows are components of $\eta$.
We have just shown that the right face is a pullback.
The front and back faces are pullbacks by the definition of $\Pi$.
And the bottom face is a pullback by the LCCC semantics of intuitionistic dependent type theory.
\begin{gather*}
\begin{gathered}
\begin{tikzpicture}[x=25pt,y=25pt]
\coordinate (O) at (0,0);
\coordinate (r) at (6,0);
\coordinate (d) at (0,-4.5);
\coordinate (dr) at (2,-1.5);
\node (A1) [inner sep=0.25em] at (O) {$\Pi_{\und{\interp{\sigma \Phi'}}, \interp{\sigma A}} \interp{\sigma B}$};
\node (A2) [inner sep=0.25em] at ($ (A1) + (r) $) {$\interp{\sigma A} \multimap_{\und{\interp{\sigma \Phi'}}} \interp{\sigma B}$};
\node (B1) [inner sep=0.25em] at ($ (A1) + (dr) $) {$\Pi_{\und{\interp{\Phi'}}, \interp{A}} \interp{B}$};
\node (B2) [inner sep=0.25em] at ($ (B1) + (r) $) {$\interp{A}\multimap_{\und{\interp{\Phi'}}} \interp{B}$};
\node (C1) [inner sep=0.25em] at ($ (A1) + (d) $) {$\und{\Pi}_{\und{\interp{\sigma \Phi'}}, \und{\interp{\sigma A}}} \und{\interp{\sigma B}}$};
\node (C2) [inner sep=0.25em] at ($ (C1) + (r) $) {$\und{\interp{\sigma A}} \Rightarrow_{\und{\interp{\sigma \Phi'}}} \und{\interp{\sigma B}}$};
\node (D1) [inner sep=0.25em] at ($ (C1) + (dr) $) {$\und{\Pi}_{\und{\interp{\Phi'}}, \und{\interp{A}}} \und{\interp{B}}$};
\node (D2) [inner sep=0.25em] at ($ (D1) + (r) $) {$\und{\interp{A}} \Rightarrow_{\und{\interp{\Phi'}}} \und{\interp{B}}$};
\draw [->] (A1) -- (A2);
\draw [->,dotted] (A1) -- (B1);
\draw [->] (A2) -- (B2);
\draw [->] (B1) -- (B2);
\draw [->] (A1) -- (C1);
\draw [->] (A2) -- (C2);
\draw [->] (B1) -- (D1);
\draw [->] (B2) -- (D2);
\draw [->] (C1) -- (C2);
\draw [->] (C1) -- (D1);
\draw [->] (C2) -- (D2);
\draw [->] (D1) -- (D2);
\end{tikzpicture}
\end{gathered}
\end{gather*}
By diagram chase, the front pullback gives the unique dotted arrow making the diagram commute.
Then, because the back, bottom, and front faces are pullbacks, the pullback lemma implies that the top face is also a pullback.
Moreover, since the top, right, and bottom faces are pullbacks, the left face is also a pullback.
Compose this pullback with the following pullback (which we have by the LCCC semantics) and we are done.
\[
    \begin{tikzcd}
	\und{\Pi}_{\und{\interp{\sigma \Phi'}}, \und{\interp{\sigma A}}} \und{\interp{\sigma B}}
	\arrow[r]
	\arrow[dr, phantom, "\LRcorner", very near start]
	\arrow[d]
	&
	\und{\Pi}_{\und{\interp{\Phi'}}, \und{\interp{A}}} \und{\interp{B}}
	\arrow[d]
	\\
	\und{\interp{\sigma \Phi'}}
	\arrow[r, "\und{\interp{R}}^\dagger"]
	&
	\und{\interp{\Phi'}}
    \end{tikzcd}
  \qedhere
\]

  \end{itemize}
\end{proof}

\begin{proposition}[Value substitution]
  \label{sub:value}
  If $\Gamma_1, x:_k A, \Gamma' \vdash M : B[x]$ and $\Gamma_2 \vdash V : A$, then
  $\Gamma_1+ k\Gamma_2, [\Sh(V)/x]\Gamma' \vdash [V/x]M : B[\Sh(V)]$.
  Semantically, suppose we have
  \[\interp{M}: \interp{\Gamma_1, x:_k A, \Gamma'} \to \interp{B},\]
  \[\interp{V}^k : \interp{\Gamma_2}^k \to \interp{A}^k,\]
  \[\interp{[V/x]M} : \interp{\Gamma_1+k\Gamma_2, [\Sh(V)/x]\Gamma'} \to \interp{B[\Sh(V)]}.\]
  Then $\interp{[V/x]M} = u$, where $u$ is the unique morphism in the diagram below.
  \[
	\begin{tikzcd}
	  \interp{\Gamma_1+k\Gamma_2, [\Sh(V)/x]\Gamma'} \arrow[r, "\interp{V}^{**}"]
	  \arrow[dd, bend right = 60]
	  \arrow[d, dashed, "u"]&
	  \interp{\Gamma_1, x:_k A, \Gamma'} \arrow[d, "\interp{M}"]\\
	  \interp{B[\Sh(V)]} \arrow[r] \arrow[d]
	  \arrow[dr, phantom, "\LRcorner\qquad", very near start] &  \interp{B} \arrow[d]
	  \\
	  \und{\interp{\Gamma_1 + k\Gamma_2, [\Sh(V)/x]\Gamma'}}
	  \arrow[r, "\und{\interp{V}^{**}}"]
	  & \und{\interp{\Gamma_1, x: A, \Gamma'}}
	\end{tikzcd}
  \]
  Note that we have
  \[\interp{V}^\dagger : \interp{[\Sh(V)/x]\Gamma'} \to \und{\interp{\Gamma'}},\]
  \[\interp{V}^{*} = \interp{\Gamma_1} \depotimes \interp{V}^k : \interp{\Gamma_1 + k\Gamma_2} \to \interp{\Gamma_1, x :_k A},\]
  \[\interp{V}^{**} = \interp{V}^{*}\depotimes \interp{V}^\dagger : \interp{\Gamma_1+k\Gamma_2, [\Sh(V)/x]\Gamma'} \to \interp{\Gamma_1, x:_k A, \Gamma'}.\]
\end{proposition}

\section{A specification of dependently typed Proto-Quipper}
\label{spec}
\begin{definition}[Syntax]
  \

  \begin{tabular}{l}
    \textit{Types} $A, B, C\  ::=  \Unit \mid \Qubit \mid \Bit \mid \Nat \mid \Vect{A} R \mid \List A$

    \\
    \quad ${} \mid (x : A) \multimap B[x] \mid (x : A) \otimes B[x] \mid {!}A \mid \Circ(S, U) \mid (x : P_1) \to P_2[x]$  \\

    \textit{Terms} $M, N \ ::= c \mid \ell \mid \ x \mid \lambda x . M \mid M N \mid \unitt \mid (a, \CC, b) \mid \mathsf{apply}(M, N) \mid \force M$ \\

    \quad ${} \mid \forceprime M \mid \liftt M \mid \boxt_U M \mid (M, N) \mid \lett (x, y) = N \tin M  \mid \lambda' x . M \mid M @ N$ \\

    \quad ${} \mid \mathsf{apply}'(M, N)$ 

    \\
    
    \textit{Parameter terms} $R ::= \ c \mid x \mid \lambda' x . R\mid R_1 @ R_2 \mid \unitt \mid (a, \CC, b) \mid \mathsf{apply}'(R_{1}, R_{2})$

    \\
    \quad ${}\mid \boxt_S R \mid \forceprime \ R \mid \liftt M \mid (R_1, R_2)\mid \lett (x, y) = R_2 \tin R_2$

    \\
    
    \textit{Parameter types} $P ::= \Unit \mid \Nat\mid P_1 \otimes P_2 \mid {!}A \mid \Circ(S, U) $
    \\
   \quad ${}\mid \Vect{P} R \mid \List P \mid (x : P_1) \to P_2[x]$

    \\
    
    \textit{Simple types} $S, U \ ::= \Unit \mid \Qubit\mid S \otimes U \mid \Vect{S} R$
    \\
    \textit{Simple Terms} $a, b, c \ ::= \ell  \mid \unitt\mid (a, b) \mid \VNil \mid \VCons a b$

    \\
    \textit{Counts} $k \ ::= 0 \mid 1 \mid \omega$
    \\

    \textit{Label Contexts} $\Sigma \ ::= \cdot \mid \ell :_{k} \Qubit, \Sigma \mid \ell :_{k} \Bit, \Sigma$
    \\
    
    \textit{Circuits} $\CC, \DD : \Sigma \to \Sigma'$

    \\

    \textit{Contexts} $\Gamma \ ::= \cdot \mid x :_k A, \Gamma \mid \ell :_{k} \Qubit, \Gamma$
    \\
    \text{\quad where $k=\omega$ only if $A$ is a parameter type.}

    \\
    \textit{Parameter context} $\Phi\ ::= \cdot \mid x :_k A, \Phi \mid \ell:_0 \Qubit, \Phi$
    \\
    \text{\quad where $k\in\s{1,\omega}$ only if $A$ is a parameter type.}

    \\
    \textit{Values}   $V ::= \unitt \mid x \mid \ell \mid \lambda x. M \mid \lambda' x. R \mid \liftt M \mid (a, \CC, b)$
  \end{tabular}
\end{definition}

The only items we have not mentioned before are the following. The
symbol $c$ ranges over constants of the language, which include the
constructors and eliminators of the data types $\Nat$, $\List A$, and
$\Vect A n$, including $\Zero$, $\Succ$, $\Nil$, $\Cons{}$, etc. We
also introduce a notion of simple terms, which are just the closed
values of simple types. We have already seen these used in boxed
circuits $(a,\CC,b)$.

\begin{definition}[Kinding]
  \[
    \begin{array}{c@{\qquad}c}
      \infer{\vdash \Nat | \Qubit | \Bit : * }{}
      &
        \infer{\Phi \vdash \Vect{\,(\Nat | \Qubit | \Bit)} R : *}
        {\Phi \vdash R : \Nat}
      \\
      \\
        \infer{\Phi \vdash (x : A) \multimap B[x] : *}{\Phi, x :_k \Sh(A) \vdash B[x] : *}
      &
        \infer{\Phi \vdash (x : P_1) \to P_2[x] : *}{\Phi, x : P_1 \vdash P_2[x] : *}                  
      \\
      \\

      \infer{\Phi \vdash (x : A) \otimes B[x] : *}{\Phi, x :_k \Sh(A) \vdash B[x] : *}

      &
      \infer{\Phi \vdash {!}A : *}{\Phi \vdash A : *}

      \\
      \\
        \infer{\Phi \vdash \Circ(S,U) : *}{\Phi \vdash S : * & \Phi \vdash U : *}
    \end{array}
  \]
\end{definition}

\begin{definition}[Typing]  
  \[ \small
    \begin{array}{cccc}

      \infer[(\textit{var})]{0\Gamma, x :_1 A, 0\Gamma' \vdash x : A}{0\Gamma, x :_1 A,  0\Gamma' \vdash }

      &
        \infer[(\textit{unit})]{\Phi \vdash \unitt : \Unit}{}

      \\
      \\

      \infer[(\textit{lam})]{\Gamma \vdash \lambda x . M : (x : A) \multimap B[x]}{\Gamma, x :_k A \vdash M : B[x] & k = 0 \Rightarrow A = P}

                                                                                                                   &

                                                                                                                     \infer[(\textit{app})]{\Gamma_1 + \Gamma_2 \vdash M N : B[\Sh(N)]}{\Gamma_1 \vdash M : (x : A) \multimap B[x] &
                                                                                                                                                                                                                                     \Gamma_2 \vdash N : A}

      \\
      \\
      
      \infer[(\textit{lift})]{\Phi \vdash \liftt M : {!}A}{\Phi \vdash M : A}
      
      &
        
        \infer[(\textit{force})]{\Gamma \vdash \force M : A}{\Gamma \vdash M : {!}A} 
      \\
      \\

      \infer[(\textit{box})]{\Gamma \vdash \boxt_S M : \Circ(S,U)}{\Gamma \vdash M : {!}(S \multimap U)}
      &
      \infer[(\textit{apply})]{\Gamma_{1} + \Gamma_{2} \vdash \mathsf{apply}(M, N) : S_{2}}{\Gamma_{1} \vdash M : \Circ(S_1,S_2) & \Gamma_{2} \vdash N : S_{1}}
      \\
      \\
      
      \infer[(\textit{pair})]{\Gamma_1 + \Gamma_2 \vdash (M, N) :  (x : A) \otimes B[x]}{\Gamma_1 \vdash M : A
      & \Gamma_2 \vdash N : B[\Sh(M)]} 
      &                                   
        \infer[(\textit{let})]
        {\Gamma_1+ \Gamma_2 \vdash \lett (x, y) = M \tin N : C}
        {
        \begin{array}{l}
          \Gamma_1 \vdash M : (x : A) \otimes B[x] \\
          \Gamma_2, x :_{k_1} A, y :_{k_2} B[x] \vdash N : C\\
          k_1 = 0 \Rightarrow A = P_1, k_2 = 0 \Rightarrow B[x] = P_2[x]
        \end{array}}

      \\
      \\
      \infer[(\textit{label})]{0\Gamma, \ell :_1 \Qubit | \Bit, 0\Gamma' \vdash \ell : \Qubit | \Bit}
      {0\Gamma, \ell :_1 \Qubit | \Bit, 0\Gamma' \vdash }

      &
      \infer[(\textit{circ})]{\Phi \vdash (a, \CC, b) :  \Circ(S, U)}{\Sigma_1 \vdash a : S & \Sigma_2 \vdash b : U & \CC : \Sigma_1 \to \Sigma_2}
      \\
      \\
      \infer[(\textit{lam}')]{\Phi \vdash \lambda' x . R : (x : P_1) \to P_2[x]}
      {\Phi, x : P_1 \vdash R : P_2[x]}
      &
      \infer[(\textit{app}')]{\Phi \vdash R_1 @ R_2 : B[R_2]}{\Phi \vdash R_1 : (x : P_1) \to P_2[x] & \Phi \vdash R_2 : P_1}
                                                                                                         
      \\
      \\
      \infer[(\textit{force}')]{\Phi \vdash \forceprime \ R :  \Sh(A)}{\Phi \vdash R : {!}A}
      &
        \infer[(\textit{apply}')]{\Phi \vdash \mathsf{apply}'(R_{1}, R_{2}) : \Sh(S_{2})}{\Phi \vdash R_{1} : \Circ(S_1,S_2) & \Phi \vdash R_{2} : \Sh(S_{1})}        
    \end{array}
  \]
\end{definition}

\begin{definition}[Simple type inhabitation]
  We define a function $\gen$ that generates a globally fresh
  inhabitant for a simple type. It fails if the input is not a simple type.
  
  $\gen(\Unit) = \unitt$

  $\gen(\Qubit) = \ell$, where $\ell$ is fresh.

  $\gen(\Bit) = \ell$, where $\ell$ is fresh.

  $\gen(S \otimes U) = (\gen(S) , \gen(U))$.

  $\gen(\Vect S \Zero) = \VNil$.

  $\gen(\Vect S (\Succ n)) = \VCons {\,(\gen(S))} {(\gen(\Vect S n))}$.
\end{definition}

\noindent So for example, $\gen(\Vect{\Qubit} 3)$ will give a vector of fresh
labels of length $3$.

\begin{definition}[Appending circuits]
  We define the circuit appending operation $\append : (\CC, c, (a, \CC', b)) \mapsto \DD$, where the circuit $\DD$ is
  obtained by connecting the input interface $a$ of $\CC'$ to a subset of outputs (exposed by the interface $c$) of $\CC$.
  Thus $\inn(\DD) = \inn(\CC)$ and $\outt(\DD) = \Sigma_{b}, \Sigma'$, where $\Sigma' = \outt(\CC) - \Sigma_c$. 
\end{definition}

We write $\Sigma_c$ to mean $\Sigma_c \vdash c : U$ for some simple type $U$.
The above circuit appending function will fail if the $a$ and $c$ have different types.
Type preservation will ensure that a well-typed program
will always be able to generate a well-formed circuit. 

\begin{definition}[Evaluation rules]
  \label{op:sem}
  \[ \footnotesize
    \begin{array}{cccc}
      \infer{(\CC, V)  \Downarrow (\CC, V)}{}
      &
        \infer{(\CC, M N)  \Downarrow (\CC''', N')}
        {
        \begin{array}{ll}
          (\CC, M) \Downarrow (\CC', \lambda x. M') &
                                                      (\CC', N) \Downarrow (\CC'', V) \\
          (\CC'', [V/x]M') \Downarrow (\CC''', N')
        \end{array}}

      \\
      \\
        \infer{(\CC_1, \mathsf{apply}(M, N))  \Downarrow (\CC_4, b)}
        {
        \begin{array}{c}
          (\CC_1, M) \Downarrow (\CC_2, (a, \DD, b)) \\
          (\CC_2, N) \Downarrow (\CC_3, a') \\
          \CC_4 = \append(\CC_3, a', (a, \DD, b))
        \end{array}
        }

      &
        \infer{(\CC, \force M)  \Downarrow (\CC'',N)}
        {(\CC, M) \Downarrow (\CC', \liftt M') &
                                                 (\CC', M') \Downarrow (\CC'', N)}
      \\
      \\
      \infer{(\CC, \forceprime R)  \Downarrow (\CC'',R')}
      {(\CC, R) \Downarrow (\CC', \liftt M) &
                                              (\CC', \Sh(M)) \Downarrow (\CC'', R')}
                                                    &
                                                      
        \infer{(\CC_1, \mathsf{apply}'(R_{1}, R_{2}))  \Downarrow (\CC_3, \Sh(b))}
        {
        \begin{array}{c}
          (\CC_1, R_{1}) \Downarrow (\CC_2, (a, \DD, b)) \\
          (\CC_2, R_{2}) \Downarrow (\CC_3, R_{2}') \\
        \end{array}
        }

      \\
      \\

      \infer{(\CC, R_1 @ R_2)  \Downarrow (\CC''', R')}
      {
      \begin{array}{ll}
        (\CC, R_1) \Downarrow (\CC', \lambda' x. R_1') &
                                                         (\CC', R_2) \Downarrow (\CC'', V) \\
        (\CC'', [V/x]R_1') \Downarrow (\CC''', R')
      \end{array}}
      &
        
        \infer{(\CC, \boxt_{S[R]} M) \Downarrow (\CC',  (a,\DD,b )) }{
        \begin{array}{ll}
          (\CC, M)\Downarrow (\CC', \liftt \ M')  & (\id_{I}, R) \Downarrow (\_, V) \\
          (\id_a, M'\ a) \Downarrow (\DD, b) & \\
          a = \gen(S[V]) & \FV(S[V]) = \emptyset
        \end{array}
                           }
                           
      \\
      \\

        \infer{(\CC, \lett (x, y) = N \tin M) \Downarrow (\CC'', N') }{(\CC, N)\Downarrow (\CC', (V_1, V_2)) & (\CC' , [V_2/y]([V_1/x]M)) \Downarrow (\CC'', N')}
    \end{array}
  \]
\end{definition}

\bibliographystyle{alphaurl}
\bibliography{lindep-long}

\end{document}